\documentclass[onecolumn,draftcls]{IEEEtran}

\usepackage{lmodern}

\usepackage{amsfonts, amsmath, amssymb, mathrsfs}
\everymath{\displaystyle}

\usepackage{xcolor}

\usepackage{graphicx, float}

\usepackage[noadjust]{cite}

\DeclareMathOperator{\tr}{tr}
\newcommand{\Deg}{\ensuremath{^\circ}}

\usepackage{amsthm}
\newtheorem{prop}{Proposition}[section]

\usepackage[percent]{overpic}

\graphicspath{{./figures/}}

\begin{document}

\title{Robust Detection for Mills Cross Sonar}

\author{
Olivier Lerda$^{1}$, Ammar Mian$^{2}$, Guillaume Ginolhac$^{2}$, {\it Senior Member, IEEE}, Jean-Philippe Ovarlez$^{3,4}$, {\it Member, IEEE}, Didier Charlot$^{1}$

\thanks{$^{1}$Exail Sonar Systems division, La Ciotat, France.}
\thanks{$^{2}$LISTIC, University Savoie Mont-Blanc, Annecy, France.}
\thanks{$^{3}$ONERA - DEMR, University Paris-Saclay, Palaiseau, France.}
\thanks{$^{4}$CentraleSupelec - SONDRA, University Paris-Saclay, Gif-sur-Yvette, France.}
}

\date{}

\maketitle
\setcounter{page}{1}

\begin{abstract}
Multi-array systems are widely used in sonar and radar applications. They can improve communication speeds, target discrimination, and imaging.
In the case of a multibeam sonar system that can operate two receiving arrays, we derive new adaptive tests to improve detection capabilities compared to traditional sonar detection approaches.
To do so, we more specifically consider correlated arrays, whose covariance matrices are estimated up to scale factors, and an impulsive clutter. In a partially homogeneous environment, the 2-step Generalized Likelihood ratio Test (GLRT) and Rao approach lead to a generalization of the Adaptive Normalized Matched Filter (ANMF) test and an equivalent numerically simpler detector with a well-established texture Constant False Alarm Rate (CFAR) behavior.
Performances are discussed and illustrated with theoretical examples, numerous simulations, and insights into experimental data.
Results show that these detectors outperform their competitors and have stronger robustness to environmental unknowns.
\end{abstract}

\begin{IEEEkeywords}
Adaptive Normalized Matched Filter, Complex Elliptically Symmetric distributions, Multiple-Input Multiple-Output, Robustness, Sonar target detection, Tyler's M-estimator.\end{IEEEkeywords}

\section{Introduction}
\label{sec:introduction}

\subsection{Background and motivations}
\label{subsec:background}

Forward-Looking sonars are solutions for perceiving the underwater environment. In the context of a growing need for decision-making autonomy and navigation safety, they have become fundamental tools for understanding, anticipating obstacles and potential dangers, analyzing and identifying threats. They offer efficient results, allowing detection, tracking, and classification of surface \cite{Karoui2015}, water column \cite{Quidu2012}, or bottom targets \cite{Quidu2010}, in civil \cite{Petillot2001} or military applications such as mine warfare \cite{Braginsky2016}.

At the detection level, monovariate statistical tests under the Gaussian or non-Gaussian interference assumption, defined a priori, remain the prevalent approaches \cite{Abraham2002, Abraham2010, Abraham2009}.
Nevertheless, many works on multivariate statistics have shown great interest compared to algorithms developed from monovariate statistics in a large number of application fields. Indeed, multivariate statistics allow advanced modeling of propagation environments. By following these precepts, \cite{Kelly1986} gets a central detector with the total unknown of the noise parameters, \cite{Robey1992} first derives a detector under the assumption of a known covariance then substitutes it, through a two-step procedure, by an appropriate estimator, finally, \cite{Scharf1994 } and \cite{ Kraut2001} have shown the relevance of subspace data models in  consideration of mismatched signals for which, as an example, the target echo would not come precisely from the center of the main beam. These seminal works are now references in remote sensing, and ground or air surveillance but mainly for radar systems. 

Moreover, in the radar field, phenomenal progress has also been made in recent decades, guided by increasingly complex systems that required profound changes in concepts and processing. This is especially the case of Space-Time Adaptive Processing (STAP) for airborne radars \cite{Ward1994}, which bring considerable improvements in the ability to discriminate moving targets at very low speeds, or Multiple-Input Multiple-Output (MIMO) systems that advance detection performance \cite{Fishler2006}, \cite{Fishler2004} and resolution \cite{Lehmann2006}, \cite{Li2007} by exploiting spatial, frequency or waveform diversities. In underwater acoustics, the use of multi-ping processing is made more complicated due to the low speed of propagation of acoustic waves in the marine environment which does not allow to respect the assumption of non-moving targets. The use of multivariate multi-ping approaches then requires more elaborate migration compensation processing that is only possible if the level of coherence is sufficient between the different pings, which is not always the case (environmental fluctuation and spatial decorrelation). Thus, although some preliminary work has emerged in recent years \cite{Wei2012, Li2009, Sasi2010, Pillai2003, Pan2022, Liu2021, Tucker2011}, these methods remain little used in sonar systems.

This paper focuses on the adaptive detection of a point target by a correlated orthogonal arrays sonar system. Inspired by these previous works, we will first show that multibeam systems are perfectly adapted to multivariate formalism. We will then propose two new detectors following the GLRT, and Rao two-step approaches \cite{Xue2022, Ciuonzo2017, Sun2022}, assuming heterogeneous or partially homogeneous clutter \cite{Wang2013, Liu2014}. The performance in a Gaussian environment will first be evaluated. We will show that considering a sonar system with Mills cross arrays \cite{Mills53} leads to a better detectability of targets by reducing the clutter ridge.

Nevertheless, complex multi-normality can sometimes be a poor approximation of physics. This is the case for fitting high-resolution clutter, impulsive noise, outliers, and interference. The Complex Elliptic Symmetric (CES) distributions \cite{Ollila2012} including the well-known compound Gaussian subclass are then natural extensions allowing the modeling of distributions with heavy or light tails in radar \cite{Farina1997, Billingsley1999, Greco2004, Ovarlez2016} as in sonar \cite{Saucan2016}. Mixtures of Scaled Gaussian (MSG) distributions \cite{Hippertferrer2022} are derived and easily tractable approaches. In this context, particular covariance matrix estimators are recommended for adaptive processing such as the Tyler estimator \cite{Pascal2008} or Huber's M-estimator \cite{Mahot2013}. Their uses lead to very substantial performance gains in \cite{Chong201003}, \cite{Petrov2018}, and \cite{Ovarlez2016}. In our application, these considerations will allow us to design a new covariance matrix estimator. The performance of the detectors in a non-Gaussian impulsive environment can then be studied. On this occasion, we will show on experimental data, this estimator's interest in the robustness to corruption of training data.

\subsection{Paper organization and contributions}
\label{subsec:organization}

We use a Mills cross sonar with two receiving arrays for point-like target detection in heterogeneous (non-)Gaussian environment. The main contributions are summarized below:
\begin{itemize}
    \item[\textbullet] We propose a new covariance matrix estimator generalizing the usual Tyler’s estimator in a context of dual and correlated arrays.
    \item[\textbullet] On these assumptions, two adaptive coherent detectors are derived based on the two-step GLRT and Rao’s test.
    \item[\textbullet] The performance of these detectors, in terms of false alarm and probability of detection, is evaluated and compared. We highlight the constant false alarm rate property with respect to the power variation and the speckle covariance matrix.
    \item[\textbullet] The contribution of a sonar architecture with orthogonal arrays in the spatial reduction of the clutter ridge is shown.
    \item[\textbullet] Experimental results demonstrate the usefulness of the proposed models and detection schemes on real sonar data.
\end{itemize}

This paper is organized as follows: Section \ref{sec:seapix} presents a dual array sonar system and the experimental acquisition conditions on which this work is based. In Section \ref{sec:detection}, the signal model and detection problem are formalized. According to the two-step GLRT and Rao test design, coherent adaptive detectors are derived in Section \ref{sec:detectors}. The performances are evaluated, compared, and analyzed in Sections \ref{sec:results_simu} and \ref{sec:results_exp}. Conclusions are given in Section \ref{sec:conclusions}. Proofs and complementary results are provided in the Appendices. \\

\indent \textit{Notations}: Matrices are in bold and capital, vectors in bold. $\mathrm{Re}(.)$ and $\mathrm{Im}(.)$ stand respectively for real and imaginary part operators. For any matrix $\mathbf{A}$ or vector, $\mathbf{A}^T$ is the transpose of $\mathbf{A}$ and $\mathbf{A}^H$ is the Hermitian transpose of $\mathbf{A}$. $\mathbf{I} $ is the identity matrix and $\mathcal{CN}(\boldsymbol{\mu},\boldsymbol{\Gamma})$ is the circular complex Normal distribution of mean $\boldsymbol{\mu}$ and covariance matrix $\boldsymbol{\Gamma}$. $\otimes$ denotes the Kronecker product.

\section{Seapix system}
\label{sec:seapix}

\subsection{Generalities}
\label{subsec:generalities}

The SEAPIX system is a three-dimensional multibeam echosounder developed by the sonar systems division of Exail (formerly iXblue) \cite{Mosca2016}. It is traditionally used by fishing professionals  as a tool to assist in the selection of catches and the respect of quotas \cite{Corbieres2017}, by hydro-acousticians for the monitoring of stocks and morphological studies of fish shoals \cite{Tallon2021}, by hydrographers for the establishment of bathymetric and sedimentary marine charts \cite{Matte2017, Nguyen2016}.

Two uniform linear arrays of 64 elements, arranged in Mills cross, are entirely symmetric, reversible in transmission/reception, and electronically steerable. They generate transverse (i.e. across-track) or longitudinal (along-track) acoustic swaths of 120\Deg\ by 1.8\Deg, tiltable on +/-60\Deg, providing a volumic coverage of the water column.

\begin{figure}[htbp]
	\centering
	\includegraphics[width = 0.35\linewidth]{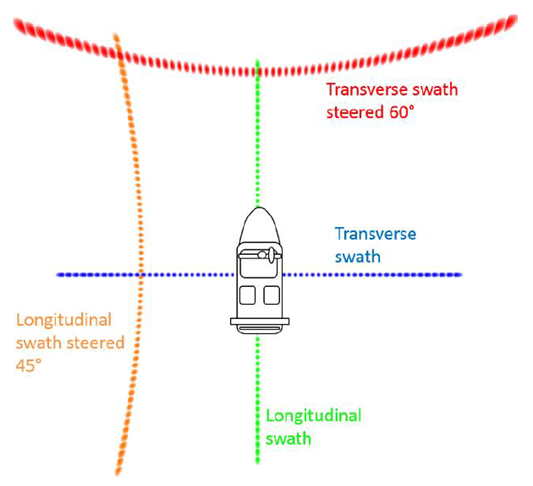}
	\includegraphics[width = 0.63\linewidth]{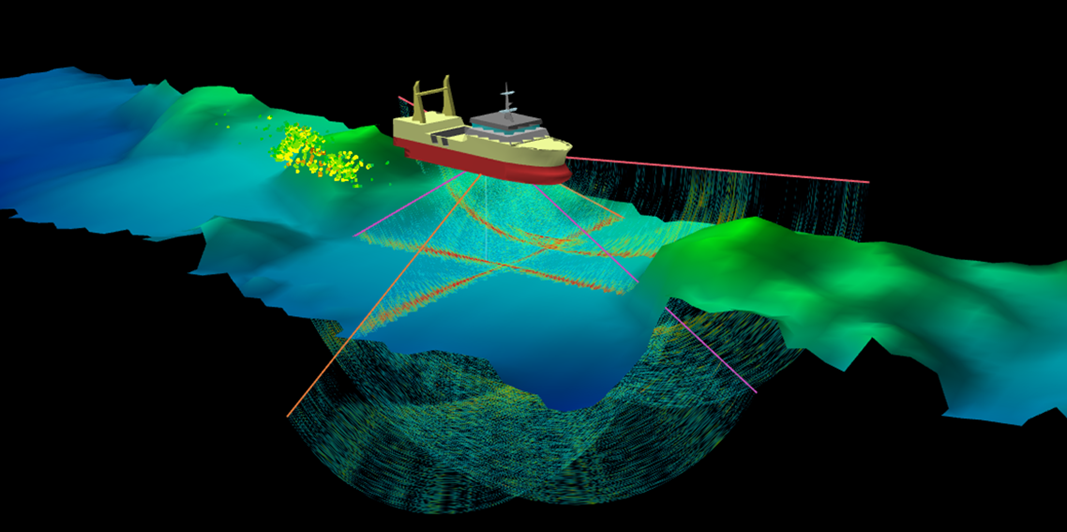}
	\caption{Multiswath capabilities seen from a schematic representation (left): transverse swath footprint is blue, a 60\Deg\ steered transverse swath is displayed in red, the longitudinal swath is green, and a 45\Deg\ steered longitudinal swath is orange. Illustration from the operator software (right): An across-track, an along-track, and a tilted along-track swath are observed, as well as an aggregation of fishes and an already constructed bathymetric map.	}
	\label{fig:multiswath}
\end{figure}

\subsection{FLS experiment}
\label{subsec:fls_experiment}

The SEAPIX system is experimented in Forward-Looking Sonar configuration for predictive target detection and identification. In this context of use, the active face is oriented in the ``forward" direction rather than toward the seabed. 

In our study, the sensor is installed on the DriX Uncrewed Surface Vehicle (USV) and inclined by 20\Deg\ according to the pitch angle to the sea surface (Figure \ref{fig:drix_and_target} left). In transmission, the vertical antenna (formerly longitudinal) generates an enlarged beam of 9\Deg\ in elevation by 120\Deg\ in azimuth. A motion stabilized, and electronically tilted firing angle allows the upper limit of the -3~dB transmit beamwidth to graze the sea surface and the lower limit to reach a 50~m depth bottom at about 300~m range. In receive, the horizontal antenna (formerly transverse) generates beams of 2\Deg\ in azimuth by 120\Deg\ in elevation and the vertical antenna (which is used again) of 2\Deg\ in elevation by 120\Deg\ in azimuth. A rigid sphere of 71 cm diameter (Figure \ref{fig:drix_and_target} right) is also immersed at 25~m depth in the middle of the water column.

\begin{figure}[htbp]
	\centering
	\includegraphics[width = 0.49\linewidth]{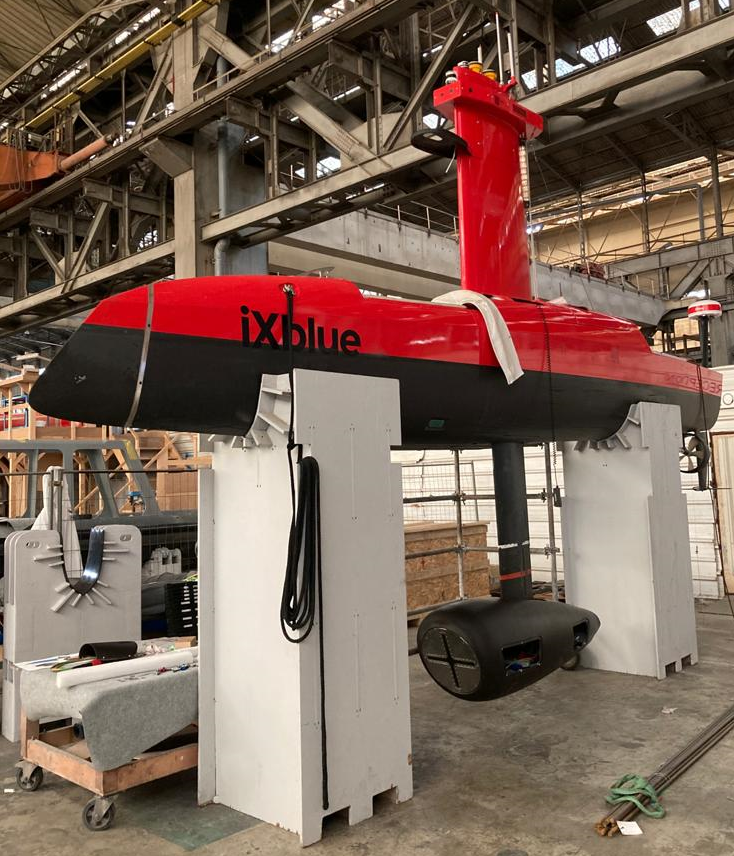}
	\includegraphics[width = 0.482\linewidth]{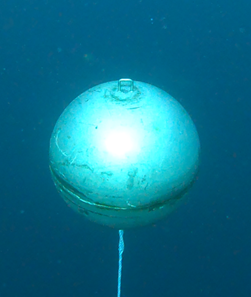}
	\caption{Exail's DriX USV (left): The cross-shaped linear arrays are visible in the gondola. The real target in open water (right): a metallic sphere of target strength $TS = -15~$dB.}
	\label{fig:drix_and_target}
\end{figure}

So after each transmission of 20~ms Linear Frequency Modulation pulses centered at 150~KHz with a 10~KHz bandwidth and a Pulse Repetition Interval of 0.5~s, the sensor signals from the two antennas are simultaneously recorded, allowing an azimuth and elevation representation of the acoustic environment (Figure \ref{fig:fls_swath}).

\begin{figure}[htbp]
	\centering
    \includegraphics[width = 1\linewidth]{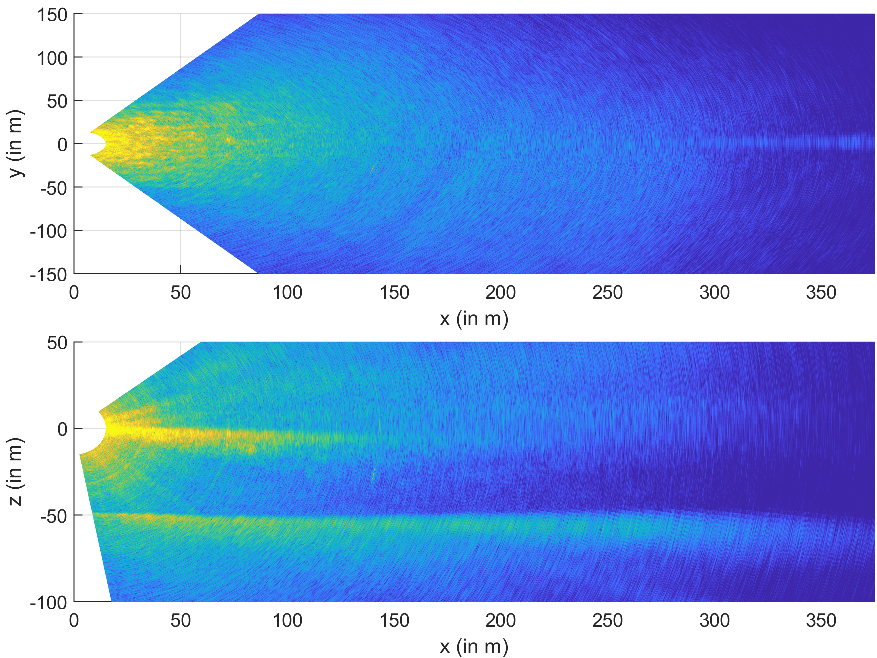}
	\caption{Azimuth (top) and elevation (bottom) view from a single ping, in 50~m depth in La Ciotat bay.}
	\label{fig:fls_swath}
\end{figure}

\subsection{Pre-processing and data format}
\label{subsec:pre_processing}

The signals from the 128 sensors provided by the sonar's embedded software are demodulated in baseband (In-Phase and Quadrature components) and decimated at 43~KHz. During reception pre-processing, the digitized samples are compensated for the time-varying gain, filtered to the bandwidth of the waveform, pulse compressed, and then decimated again to bandwidth. Finally, a ping dataset is a matrix of about 6000 range bins from 15~m to 400~m, by 64 sensors, by two arrays.

\section{Detection Schemes}
\label{sec:detection}

\subsection{Data model for a single array}

At each time, we have two digitized synchronized complex data vectors of $m=64$ elements, called ``snapshots'', which can be written as (by omitting the temporal parameterization):
\begin{equation}
	\mathbf{x}_{i} = 
	\left[
	x_{i,1}, x_{i,2}, \cdots,  x_{i,m}
	\right]^T\, ,
\end{equation}
where $i=1,2$ is the array identifier (respectively horizontal and vertical). After the pre-processing steps, a point-like target observed on the antenna $i$ is simply:
\begin{equation}
	\mathbf{x}_{i} = \alpha_{i} \, \mathbf{p}_{i} + \mathbf{z}_{i} \, ,
\end{equation}
where $\mathbf{x}_{i} \in \mathbb{C}^m$ is the received signal, $\alpha_{i}$ is an unknown complex target amplitude, $\mathbf{p}_{i} \in \mathbb{C}^m$ stands for the known deterministic angular steering vector, $\mathbf{z}_{i} \in \mathbb{C}^m$ is a mixture of scaled Gaussian (MSG) random vector admitting the stochastic representation:
\begin{equation}
	\mathbf{z}_{i} \overset{d}{=} \sqrt{\tau_i} \, \mathbf{c}_{i}\, .
	\label{eq:def_CG}
\end{equation}

The \emph{texture} $\tau_i$ is an unknown positive deterministic scalar parameter, presumably different for each range bin (i.e. for all time samples). The \emph{speckle} $\mathbf{c}_{i} \sim \mathcal{CN}(\mathbf{0}, \sigma_i^2 \, \mathbf{M}_{ii}) \in \mathbb{C}^m$ is a complex circular Gaussian random vector whose covariance matrix $\mathbf{M}_{ii}$ is known up to a scaling factor $\sigma_i^2$. The term \emph{speckle} should be understood in the sense of CES statistics rather than as the result of a sum of contributions from reflections in a resolution cell. This model is strongly related to the class of compound Gaussian distributions \cite{Ollila2012}, which assumes a speckle-independent random texture with a given density $p_\tau$. Considering deterministic texture allows building estimators and detectors which are completely independent of the knowledge of the texture PDF. In this way, detectors are no longer specific to a single and unique distribution, which can be seen as an additional source of robustness. Furthermore, the MSG distribution is more robust than the Gaussian one because the relative scaling between samples allows flexibility in the presence of heterogeneities, such as impulsive noise, outliers, and inconsistent data. This model explicitly allows considering the power fluctuation across range bins, especially for heavy-tailed clutter distributions.

The detection problem is written as a binary hypothesis test:
\begin{equation}
\left\lbrace
	\begin{array}{ll}
		H_0 : \mathbf{x}_i = \mathbf{z}_i & ; \quad \mathbf{x}_{i,k} = \mathbf{z}_{i,k} \quad k = 1 \ldots K\, , \\
		H_1 : \mathbf{x}_i = \alpha_{i} \, \mathbf{p}_{i} + \mathbf{z}_{i} & ; \quad \mathbf{x}_{i,k} = \mathbf{z}_{i,k} \quad k = 1 \ldots K \, .
	\end{array}
\right.\
\label{eq:test_hypotheses_bin_1_antenne}
\end{equation}

In \eqref{eq:test_hypotheses_bin_1_antenne} it is assumed that $K \geq m$ independent and identically distributed (i.i.d.) signal-free secondary data $\mathbf{x}_{i,k} \in \mathbb{C}^m$ are available under both hypotheses for background parameters estimation. We recall that $\mathbf{z}_{i,k} \overset{d}{=} \sqrt{\tau_{i,k}} \, \mathbf{c}_{i,k}$, $\mathbf{c}_{i,k} \sim \mathcal{CN}(\mathbf{0}, \mathbf{M}_{ii})$.

Conditionally to the unknown deterministic texture, the densities of $\mathbf{x}_{i}$ under $H_0$ and $H_1$ are Gaussian:
\begin{equation}
	\begin{array}{ll}
		p_{\mathbf{x}_{i}}(\mathbf{x}_{i}; H_0) = \cfrac{1}{\pi^m \tilde{\sigma}_i^{2m} \lvert \mathbf{M}_{ii} \rvert} \exp \left( - \cfrac{\mathbf{x}_i^H \mathbf{M}_{ii}^{-1} \mathbf{x}_{i}}{\tilde{\sigma}_i^{2}} \right) \, ,\\
		p_{\mathbf{x}_{i}}(\mathbf{x}_{i}; H_1) = \\ 
		\cfrac{1}{\pi^m \tilde{\sigma}_i^{2m} \lvert \mathbf{M}_{ii} \rvert} \exp \left( - \cfrac{{\left( \mathbf{x}_{i} - \alpha_{i} \mathbf{p}_{i}\right)}^H \mathbf{M}_{ii}^{-1} \left( \mathbf{x}_{i} - \alpha_{i} \mathbf{p}_{i} \right) }{\tilde{\sigma}_i^{2}} \right)\, ,
	\end{array}
\end{equation}
where $\tilde{\sigma}_i = \sigma_i \, \sqrt{\tau_i}$.

\subsection{Data model for the two arrays}

If we consider the two antennas, this model can be written more appropriately:
\begin{equation}
\left\lbrace
	\begin{array}{ll}
	H_0 : \mathbf{x} = \mathbf{z} & ; \quad \mathbf{x}_k = \mathbf{z}_k \quad k = 1 \ldots K \, ,\\
	H_1 : \mathbf{x} = \mathbf{P} \, \boldsymbol{\alpha} + \mathbf{z} & ; \quad \mathbf{x}_k = \mathbf{z}_k \quad k = 1 \ldots K \, ,
	\end{array}
\right. 
\label{eq:test_hypotheses_bin}
\end{equation}
where $\mathbf{x} = \begin{bmatrix} \mathbf{x}_1 \\ \mathbf{x}_2 \end{bmatrix} \in \mathbb{C}^{2m}$ is the concatenation of the two received signals, $\boldsymbol{\alpha} = \begin{bmatrix} \alpha_1 \\ \alpha_2 \end{bmatrix} \in \mathbb{C}^{2}$ is the vector of the target amplitudes and $\mathbf{z} = \begin{bmatrix} \mathbf{z}_1 \\ \mathbf{z}_2 \end{bmatrix} \in \mathbb{C}^{2m}$ is the additive clutter. The matrix $\mathbf{P} = \begin{bmatrix} 
\mathbf{p}_1 & \mathbf{0} \\ 
\mathbf{0}   & \mathbf{p}_2 \\
\end{bmatrix} \in \mathbb{C}^{2m \times 2}$ contains the steering vectors. $\left\{\mathbf{x}_k = \begin{bmatrix} \mathbf{x}_{1,k} \\ \mathbf{x}_{2,k} \end{bmatrix}\right\}_{k\in [1, K]} \in \mathbb{C}^{2m}$ for $K \geq 2\,m$ are i.i.d. signal-free secondary data.

This formulation allows considering the correlation between sensors of arrays. The covariance is $\mathbf{M} = \begin{bmatrix} \mathbf{M}_{11} & \mathbf{M}_{12} \\ \mathbf{M}_{21} & \mathbf{M}_{22} \end{bmatrix}$, with $\mathbf{M}_{ii}$ the block-covariance matrix of array $i$, and $\mathbf{M}_{ij}=\mathbf{M}_{ji}^H$ the cross-correlation block of array $i$ and $j$. The partitioned inverse covariance will be denoted as $\mathbf{M}^{-1} = \begin{bmatrix} \mathbf{M}_{11}^{-1} & \mathbf{M}_{12}^{-1} \\ \mathbf{M}_{21}^{-1} & \mathbf{M}_{22}^{-1} \end{bmatrix}$.

We further assume that the covariance is known, or estimated, up to two scalars, possibly different on each array. These scalars are conditioning the $\mathbf{M}_{ii}$ block-covariance, but also all the cross-correlations blocks associated with the array $i$. We can therefore write:
\begin{equation}
	\widetilde{\mathbf{C}} = \widetilde{\boldsymbol{\Sigma}} \,\mathbf{M} \, \widetilde{\boldsymbol{\Sigma}} \, ,
\end{equation}
with $\widetilde{\boldsymbol{\Sigma}} = \begin{bmatrix} 
\tilde{\sigma}_1 \, \mathbf{I}_m & \mathbf{0} \\ 
\mathbf{0} & \tilde{\sigma}_2 \, \mathbf{I}_m \\
\end{bmatrix}$ = $\begin{bmatrix} 
\tilde{\sigma}_1	 & 0 \\ 
0 & \tilde{\sigma}_2 \\
\end{bmatrix} \otimes \mathbf{I}_m$ the unknown diagonal matrix of scalars $\sigma_i$ and $\tau_i$. We recap that the $\sigma_i$ parameter drives the partial homogeneity of the data (i.e. the difference in scale factor between the covariance matrices of the primary and secondary data) and $\tau_i$ drives the non-Gaussianity of the data (i.e. the power variation of the observations over time).

In this model, each array, although correlated, has a possibly distinct texture and an unknown scaling factor on the covariance matrix which may also be dissimilar. It would therefore become entirely possible to model, as an example, a first array whose observations would be Gaussian ($\tau_1 = 1$) and homogeneous ($\sigma_1 = 1$), and a second array whose observations would be \emph{K}-distributed (for which $\tau_2$ is a realization of a Gamma distribution) and non-homogeneous ($\sigma_2 \neq 1$). The PDFs under each hypothesis can be rewritten as:
\begin{equation}
	\begin{array}{ll}
		p_{\mathbf{x}}(\mathbf{x}; H_0) = \cfrac{1}{\pi^{2m} \lvert \widetilde{\mathbf{C}} \rvert} \exp \left( - {\mathbf{x}}^H \widetilde{\mathbf{C}}^{-1} \mathbf{x} \right) \, ,\\
		p_{\mathbf{x}}(\mathbf{x}; H_1) = \cfrac{1}{\pi^{2m} \lvert \widetilde{\mathbf{C}} \rvert} \exp \left( - {\left( \mathbf{x} - \mathbf{P} \boldsymbol{\alpha} \right)}^H \widetilde{\mathbf{C}}^{-1} \left( \mathbf{x} - \mathbf{P} \boldsymbol{\alpha} \right) \right) .
	\end{array}
\label{eq:PDF_2_ss_antennes}
\end{equation}

\section{ROBUST DETECTORS}
\label{sec:detectors}

We discuss the derivation of detectors using the GLRT and Rao procedures. Following a two-step approach, the covariance matrix $\mathbf{M}$ will first be assumed known, and then replaced by an appropriate estimator.

\subsection{Detectors' derivation with $\mathbf{M}$ known (step-1)}
\label{subsec:detectors_step_1}

\subsubsection{Generalized Likelihood Ratio Test}
\label{subsubsec:GLRT}

The Generalized Likelihood Ratio Test (GLRT) design methodology proposes to solve the detection problem from the ratio of the probability densities function under $H_1$ and $H_0$, substituting the unknown parameters with their maximum likelihood estimates:
\begin{equation}
	L_G(\mathbf{x}) = \cfrac{\underset{\boldsymbol{\alpha}}{\max} \, \underset{\widetilde{\boldsymbol{\Sigma}}}{\max} \, p_{\mathbf{x}}\left(\mathbf{x}; \boldsymbol{\alpha},\widetilde{\boldsymbol{\Sigma}}, H_1\right)}{\underset{\widetilde{\boldsymbol{\Sigma}}}{\max} \, p_{\mathbf{x}}\left(\mathbf{x}; \widetilde{\boldsymbol{\Sigma}}, H_0\right)} \, .
\end{equation}

\begin{prop} 
The GLRT for the hypothesis test \eqref{eq:test_hypotheses_bin} is given by:
\begin{equation}
	L_G(\mathbf{x}) = \cfrac{\hat{\tilde{\sigma}}_{1_0} \, \hat{\tilde{\sigma}}_{2_0}}{\hat{\tilde{\sigma}}_{1_1} \, \hat{\tilde{\sigma}}_{2_1}}\, ,
	\label{eq:glrt_detector}
\end{equation}
where
\begin{equation*}
\left( \hat{\tilde{\sigma}}_{1_0}  \hat{\tilde{\sigma}}_{2_0} \right) ^2 = 
\cfrac{\emph{Re} \left( \mathbf{x}_1^H \mathbf{M}_{12}^{-1} \mathbf{x}_2 \right)}{m}  
+ \sqrt{\cfrac{\mathbf{x}_1^H \mathbf{M}_{11}^{-1} \mathbf{x}_1}{m}  \cfrac{\mathbf{x}_2^H \mathbf{M}_{22}^{-1} \mathbf{x}_2}{m}}  \, ,
\end{equation*}
and
\begin{equation*}
\left( \hat{\tilde{\sigma}}_{1_1} \, \hat{\tilde{\sigma}}_{2_1} \right) ^2 = 
\cfrac{\emph{Re} \left[ \mathbf{x}_1^H \left( \mathbf{M}_{12}^{-1} - \mathbf{D}_{12}^{-1} \right) \mathbf{x}_2 \right]}{m} \\
+ \sqrt{\cfrac{\mathbf{x}_1^H \left( \mathbf{M}_{11}^{-1} - \mathbf{D}_{11}^{-1} \right) \mathbf{x}_1}{m} \, \cfrac{\mathbf{x}_2^H \left( \mathbf{M}_{22}^{-1} - \mathbf{D}_{22}^{-1} \right) \mathbf{x}_2}{m}}
\end{equation*}
with $\mathbf{D}^{-1} = \mathbf{M}^{-1} \mathbf{P} \left( \mathbf{P}^H \mathbf{M}^{-1} \mathbf{P} \right) ^{-1} \mathbf{P}^H \mathbf{M}^{-1}$.
\end{prop}

\begin{proof}
See Appendix \ref{sec:appendixA} for a step-by-step derivation and Appendix \ref{sec:appendixB} for some interesting equivalences.
\end{proof}

As $\mathbf{x}_{i} = \sqrt{\tau_i} \, \mathbf{c}_{i}$ under $H_0$, it is easily shown that this detector is texture independent (i.e. it has the \emph{texture-CFAR} property). This detection test will be called \emph{M-NMF-G} in the following.

\subsubsection{Rao test}
\label{subsubsec:rao}

The Rao test is obtained by exploiting the asymptotic efficiency of the ML estimate and expanding the likelihood ratio in the neighborhood of the estimated parameters \cite{Aubry2016}. A traditional approach for complex unknown parameters is to form a corresponding real-valued parameter vector and then use the real Rao test \cite{Kay1998, Kay2016}. Specifically, rewriting the detection problem \eqref{eq:test_hypotheses_bin} as:

\begin{equation}
\left\lbrace
	\begin{array}{ll}
	H_0 : \boldsymbol{\xi}_R = \mathbf{0}, \boldsymbol{\xi}_S \\
	H_1 : \boldsymbol{\xi}_R \neq \mathbf{0}, \boldsymbol{\xi}_S \, ,
	\end{array}
\right.\
\label{eq:test_hypotheses_bin_rao}
\end{equation}
where $\boldsymbol{\xi}_R = \left[ \mathrm{Re} \left( \boldsymbol{\alpha} \right) ^T \, \mathrm{Im} \left( \boldsymbol{\alpha} \right) ^T \right]^T$ is a $4 \times 1$ parameter vector and $\boldsymbol{\xi}_S = \left[ \tilde{\sigma}_1 \, \tilde{\sigma}_2 \right]^T$ is a $2 \times 1$ vector of nuisance parameters, the  Rao test for the problem \eqref{eq:test_hypotheses_bin_rao} is:

\begin{equation}
    L_R(\mathbf{x}) =
    \left. \cfrac{\partial \ln p_{\mathbf{x}}(\mathbf{x}; \boldsymbol{\xi}_R, \boldsymbol{\xi}_S)}{\partial \boldsymbol{\xi}_R} \right|_{\begin{matrix}
    \boldsymbol{\xi}_R = \hat{\boldsymbol{\xi}}_{R_0} \\
    \boldsymbol{\xi}_S = \hat{\boldsymbol{\xi}}_{S_0}
    \end{matrix}} ^T     
    \left[ \mathbf{I}^{-1} (\hat{\boldsymbol{\xi}}_{R_0}, \hat{\boldsymbol{\xi}}_{S_0}) \right]_{\boldsymbol{\xi}_R \boldsymbol{\xi}_R}   
    \left. \cfrac{\partial \ln p_{\mathbf{x}}(\mathbf{x}; \boldsymbol{\xi}_R, \boldsymbol{\xi}_S)}{\partial \boldsymbol{\xi}_R} \right|_{\begin{matrix}
    \boldsymbol{\xi}_R = \hat{\boldsymbol{\xi}}_{R_0} \\
    \boldsymbol{\xi}_S = \hat{\boldsymbol{\xi}}_{S_0}
\end{matrix}} \, .
    \label{eq:rao_formula}
\end{equation}
The PDF $p_{\mathbf{x}}(\mathbf{x}; \boldsymbol{\xi}_R, \boldsymbol{\xi}_S)$ is given in \eqref{eq:PDF_2_ss_antennes} and parametrized by $\boldsymbol{\xi}_R$ and $\boldsymbol{\xi}_S$. $\hat{\boldsymbol{\xi}}_{R_0}$ and $\hat{\boldsymbol{\xi}}_{S_0}$ are the ML estimates of $\boldsymbol{\xi}_{R}$ and $\boldsymbol{\xi}_{S}$ under $H_0$. $\mathbf{I}(\boldsymbol{\xi}_R, \boldsymbol{\xi}_S)$ is the Fisher Information Matrix that can be partitioned as:
\begin{equation*}
    \mathbf{I}(\boldsymbol{\xi}_R, \boldsymbol{\xi}_S) = \begin{bmatrix} 
    \mathbf{I}_{\boldsymbol{\xi}_R \boldsymbol{\xi}_R}(\boldsymbol{\xi}_R, \boldsymbol{\xi}_S) & \mathbf{I}_{\boldsymbol{\xi}_R \boldsymbol{\xi}_S}(\boldsymbol{\xi}_R, \boldsymbol{\xi}_S)\\ 
    \mathbf{I}_{\boldsymbol{\xi}_S \boldsymbol{\xi}_R}(\boldsymbol{\xi}_R, \boldsymbol{\xi}_S) & \mathbf{I}_{\boldsymbol{\xi}_S \boldsymbol{\xi}_S}(\boldsymbol{\xi}_R, \boldsymbol{\xi}_S)
    \end{bmatrix} \, ,
\end{equation*}
and we have:
\begin{equation}
    \left[ \mathbf{I}^{-1} (\boldsymbol{\xi}_R, \boldsymbol{\xi}_S) \right] _{\boldsymbol{\xi}_R \boldsymbol{\xi}_R} =
    \left( 
    \mathbf{I}_{\boldsymbol{\xi}_R \boldsymbol{\xi}_R} (\boldsymbol{\xi}_R, \boldsymbol{\xi}_S) - 
    \mathbf{I}_{\boldsymbol{\xi}_R \boldsymbol{\xi}_S} (\boldsymbol{\xi}_R, \boldsymbol{\xi}_S) 
    \mathbf{I}_{\boldsymbol{\xi}_S \boldsymbol{\xi}_S}^{-1} (\boldsymbol{\xi}_R, \boldsymbol{\xi}_S) \,
    \mathbf{I}_{\boldsymbol{\xi}_S \boldsymbol{\xi}_R} (\boldsymbol{\xi}_R, \boldsymbol{\xi}_S) 
   \right) ^{-1} \, .
\end{equation}
The following proposition can be finally stated:
\begin{prop} 
The Rao test is given by:
\begin{equation}
L_R(\mathbf{x}) = 2 \, \mathbf{x}^H \hat{\widetilde{\mathbf{C}}}_0^{-1} \mathbf{P} \left( \mathbf{P}^H \hat{\widetilde{\mathbf{C}}}_0^{-1} \mathbf{P} \right)^{-1} \mathbf{P}^H \hat{\widetilde{\mathbf{C}}}_0^{-1} \mathbf{x}\, ,
\label{eq:rao_dector}
\end{equation}
where $\hat{\widetilde{\mathbf{C}}}_0 = \hat{\widetilde{\boldsymbol{\Sigma}}}_0 \, \mathbf{M} \, \hat{\widetilde{\boldsymbol{\Sigma}}}_0$ and $\hat{\widetilde{\boldsymbol{\Sigma}}}_0$ is defined in Appendix \ref{sec:appendixA}.
\end{prop}

\begin{proof}
See Appendix \ref{sec:appendixC}.
\end{proof}

As for the GLRT, \eqref{eq:rao_dector} is \emph{texture-CFAR}. This detector will be referred to as \emph{M-NMF-R} in the sequel.

\subsection{Covariance estimation and adaptive detectors (step-2)}
\label{subsec:step_2}

In practice, the noise covariance matrix is unknown and estimated using the $K$ available i.i.d. signal-free secondary data.

In Gaussian environment, in which PDFs are given by \eqref{eq:PDF_2_ss_antennes} with $\sigma_i = 1$ and $\tau_i = 1$, the MLE of $\mathbf{M}$ is the well-known Sample Covariance Matrix (SCM):
\begin{equation}
	\widehat{\mathbf{M}}_{SCM} = \cfrac{1}{K} \sum_{k=1}^K \mathbf{x}_k \, \mathbf{x}_k^H\, ,
\label{eq:cov_SCM}
\end{equation}
which is an unbiased and minimum variance estimator. 

In the presence of outliers or a heavy-tailed distribution (as modeled by a mixture of scaled Gaussian), this estimator is no longer optimal or robust. This leads to a strong performance degradation.

\begin{prop} 
In MSG environment the MLE of $\mathbf{M}$ is given by:
\begin{equation}
\widehat{\mathbf{M}}_{\text{2TYL}} = \cfrac{1}{K} \sum_{k=1}^K \widehat{\mathbf{T}}_k^{-1} \mathbf{x}_k \mathbf{x}_k^H \widehat{\mathbf{T}}_k^{-1}\, ,
\label{eq:cov_Tyler}
\end{equation}
where $\widehat{\mathbf{T}}_k = \begin{bmatrix} \sqrt{\hat{\tau}_{1_k}} & 0 \\ 0 & \sqrt{\hat{\tau}_{2_k}} \\ \end{bmatrix} \otimes \mathbf{I}_m$, $\hat{\tau}_{1_k} = t_1 + \sqrt{\cfrac{t_1}{t_2}} \, t_{12}$, $\hat{\tau}_{2_k} = t_2 + \sqrt{\cfrac{t_2}{t_1}} \, t_{12}$ and $t_1 = \cfrac{\mathbf{x}_{1,k}^H \widehat{\mathbf{M}}_{11}^{-1} \mathbf{x}_{1,k}}{m}$, $t_2 = \cfrac{\mathbf{x}_{2,k}^H \widehat{\mathbf{M}}_{22}^{-1} \mathbf{x}_{2,k}}{m}$, $t_{12} = \cfrac{\text{Re} \left( \mathbf{x}_{1,k}^H \widehat{\mathbf{M}}_{12}^{-1} \mathbf{x}_{2,k} \right)}{m}$.
\end{prop}

\begin{proof}
The demonstration is provided in Appendix \ref{sec:appendixD}.
\end{proof}

A key point is that this  estimator is independent of the textures, i.e. the power variations, on each array. It can be thought of as a \emph{multi-texture} generalization of \cite{Pascal2008}.

From a practical point of view $\widehat{\mathbf{M}}_{\text{2TYL}}$ is the solution of the recursive algorithm:
\begin{align}
	\widehat{\mathbf{T}}_k^{(n)} & = \begin{bmatrix} \sqrt{\hat{\tau}_{1_k}^{(n)}} & 0 \\ 0 & \sqrt{\hat{\tau}_{2_k}^{(n)}} \\ \end{bmatrix} \otimes \mathbf{I}_m \label{eq:Tyler_T_rec} \, ,\\
	\widehat{\mathbf{M}}_{\text{2TYL}}^{(n)} & = \cfrac{1}{K} \sum_{k=1}^K \left( \widehat{\mathbf{T}}_k^{(n-1)} \right)^{-1} \mathbf{x}_k \mathbf{x}_k^H \left( \widehat{\mathbf{T}}_k^{(n-1)} \right)^{-1} \, ,\label{eq:Tyler_M_rec}
\end{align}
where $n \in \mathbb{N}$ is the iteration number, whatever the initialization $\widehat{\mathbf{M}}_{\text{2TYL}}^{(0)}$. The convergence of recursive equations \eqref{eq:Tyler_T_rec} and \eqref{eq:Tyler_M_rec} in the estimation of $\widehat{\mathbf{M}}_{\text{2TYL}}$ is illustrated in Figure \ref{fig:conv_Tyler_MIMO} where we plot the relative error between two successive estimates of the covariance matrix as a function of the iteration number for 500 steps and $\widehat{\mathbf{M}}_{\text{2TYL}}^{(0)} = \mathbf{I}_{2m}$. The relative error is given by $\left\| \widehat{\mathbf{M}}_{\text{2TYL}}^{(n)} - \widehat{\mathbf{M}}_{\text{2TYL}}^{(n-1)} \right\| / \left\| \widehat{\mathbf{M}}_{\text{2TYL}}^{(n-1)} \right\|$
where $\left\|.\right\|$ is the Frobenius norm. The relative difference between estimates decreases with the number of iterations. From iteration 60, the accuracy becomes limited by the simulation environment. In practice, it is not necessary to go to this limit, and we notice that from iteration 20 the relative deviation becomes lower than $10^{-6}$. 

\begin{figure}[htbp]
	\centering
	\includegraphics[width = 1\linewidth]{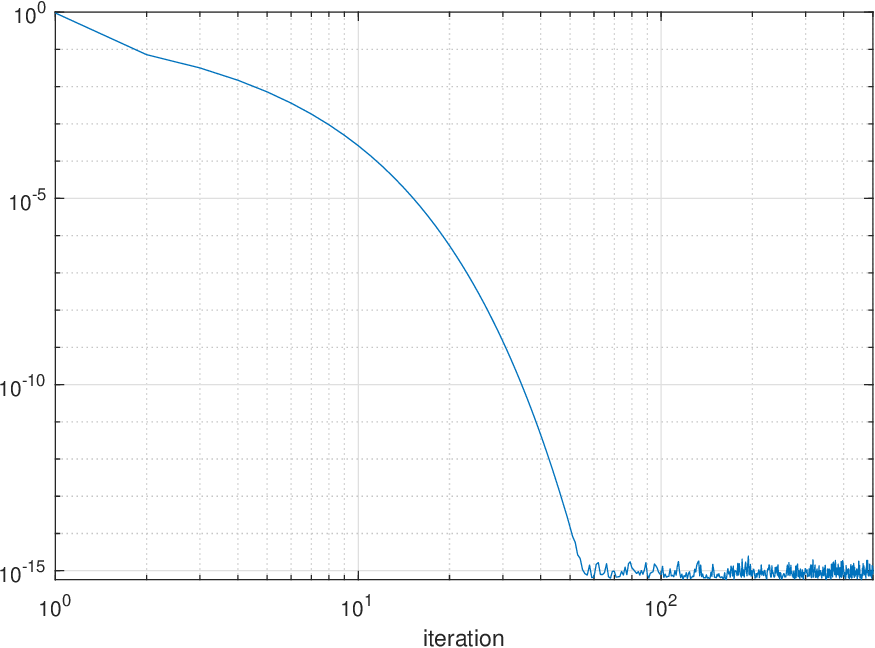}
	\caption{Relative estimation error of $\widehat{\mathbf{M}}_{\text{2TYL}}$ versus the number of iterations $n$ ($\widehat{\mathbf{M}}_{\text{2TYL}}^{(0)} = \mathbf{I}_{2m}$).}
	\label{fig:conv_Tyler_MIMO}
\end{figure}

At last, the adaptive versions of the tests \eqref{eq:glrt_detector} and \eqref{eq:rao_dector} will be simply obtained by replacing the covariance matrix $\mathbf{M}$ by an appropriate estimate: \eqref{eq:cov_SCM} or \eqref{eq:cov_Tyler} according to the environment. Those detectors will be referred as \emph{M-ANMF-G$_{SCM}$}, \emph{M-ANMF-R$_{SCM}$}, \emph{M-ANMF-G$_{TYL}$}, and \emph{M-ANMF-R$_{TYL}$}.

\section{Numerical results on simulated data}
\label{sec:results_simu}

\subsection{Performance assessment}
\label{subsec:perf_assessment}

Two correlation coefficients $\rho_1$ and $\rho_2$ $\left( 0 < \rho_1,\rho_2 < 1 \right)$, are used in the construction of a \emph{speckle} covariance matrix model defined as:
\begin{equation*}
    \left[ \mathbf{M} \right]_{jl} = \beta \, \rho_1^{\lvert j_1-l_1 \rvert} \times \rho_2^{\lvert j_2-l_2 \rvert} \, ,
\end{equation*}
where $j,l \in \left[ 1, \, 2m \right]$ are sensor numbers of coordinates $\left(j_1, j_2\right)$ and $\left(l_1, l_2\right)$ respectively and $\beta$ is a scale factor.
Thus, denoting $\mathbf{M} = \begin{bmatrix} \mathbf{M}_{11} & \mathbf{M}_{12} \\ \mathbf{M}_{21} & \mathbf{M}_{22} \end{bmatrix}$, $\mathbf{M}_{11}$ and $\mathbf{M}_{22}$ are the covariance matrices of array 1 and 2 and $\mathbf{M}_{12}=\mathbf{M}_{21}^H$ is a cross-correlation block.
In the FLS context, the block $\mathbf{M}_{11}$ is weakly correlated and Toeplitz structured (close to the identity matrix). $\mathbf{M}_{22}$ is also Toeplitz but more strongly correlated (due to the wider transmission beam). The cross-correlation blocks could be considered null under the uncorrelated arrays assumption. In our case, these will be different from zero because the arrays cross each other in their central part. This results in the general structure displayed in Figure \ref{fig:covariance}, where we visually show the adequacy of this model with the SCM covariance estimator established on real data.

\begin{figure}[htbp]
	\begin{minipage}{0.25\linewidth}
	\centering
	\includegraphics[width = 1.15\linewidth]{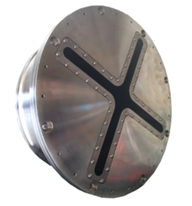}
	\end{minipage}
	\begin{minipage}{0.36\linewidth}
	\centering
	\includegraphics[width = 1\linewidth]{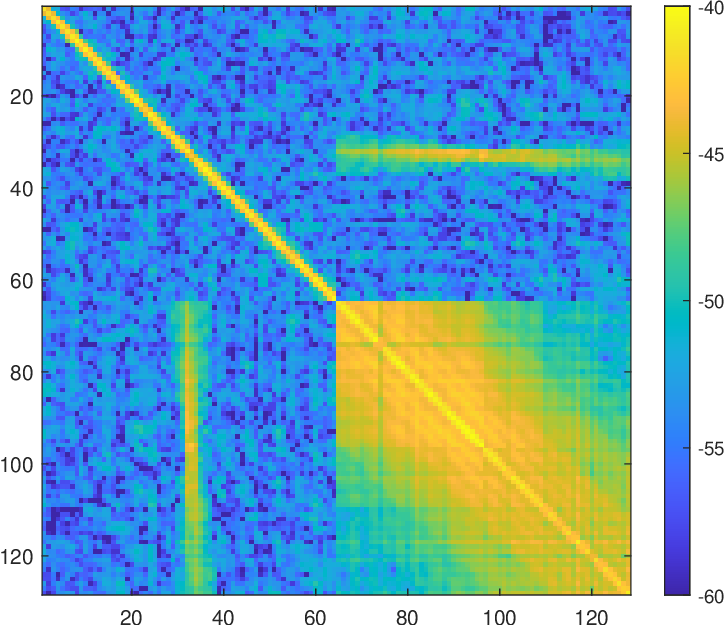}
	\end{minipage}
	\begin{minipage}{0.36\linewidth}
	\centering
	\includegraphics[width = 1\linewidth]{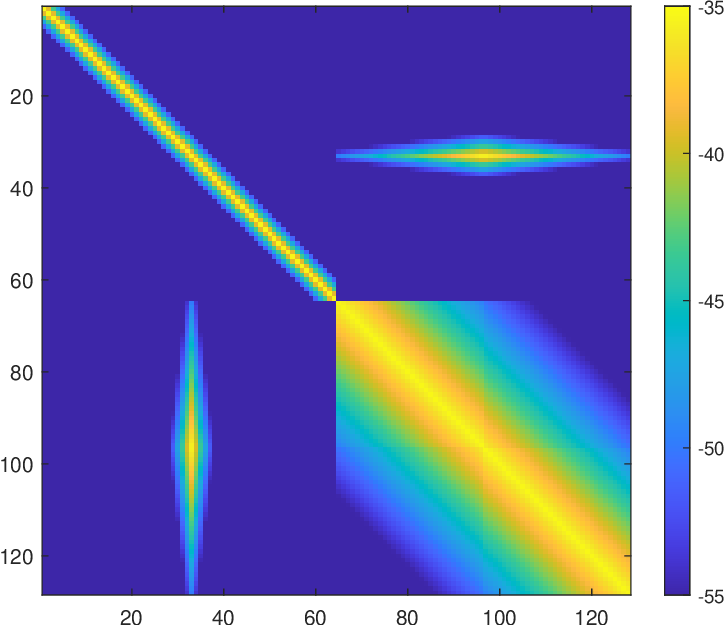}
	\end{minipage}
	\caption{Sonar arrays (left): two uniform linear arrays intersect in their centers. Empirical covariance matrix on real sonar data shown in dB (center): The SCM is estimated on a homogeneous area between 265~m and 328~m. Cross-correlation blocks are non-zero. The covariance matrix model used for the data simulation is shown in dB (right): $\beta = 3 \times 10^{-4}$, $\rho_1=0.4$, $\rho_2=0.9$.}
	\label{fig:covariance}
\end{figure}

We choose $\sigma_i = 1$. $\tau_i=1$ for Gaussian clutter and $\tau_i \sim \mathrm{Gam}(\nu,1/\nu)$ with $\nu = 0.5$ for impulsive non-Gaussian \emph{K}-distributed data (that is, the texture variables follow a gamma distribution with shape parameter 0.5 and scale parameter 2).
The PFA-threshold curves are established on the basis of 1000 realizations of random vectors.
The detection probabilities are statistically estimated by adding a synthetic narrow band far field point target with identical amplitudes on the arrays ($\alpha_1 = \alpha_2$). 10000 Monte Carlo iterations are performed and 150 target amplitudes are evaluated.

\subsection{Benchmark tests}
\label{subsec:benchmark}

The GLRT for a single array in a partially homogeneous Gaussian environment, when $\mathbf{M}_{ii}$ is known, is the \emph{Normalized Matched Filter} \cite{Scharf1994}:
\begin{equation}
\text{NMF} i(\mathbf{x}_i) = \cfrac{\lvert\mathbf{p}_i^H \mathbf{M}_{ii}^{-1} \mathbf{x}_i \rvert ^2}{\left( \mathbf{p}_i^H \mathbf{M}_{ii}^{-1} \mathbf{p}_i \right)\left( \mathbf{x}_i^H \mathbf{M}_{ii}^{-1} \mathbf{x}_i \right)} \, .
\end{equation}
Adaptive versions are obtained by substituting the covariance matrix by a suitable estimate \cite{Ovarlez2016} and will be referred to as ANMF$_{SCM} i$ in Gaussian case, or ANMF$_{TYL} i$ for the non-Gaussian case.

When the two arrays are considered, in the very favorable case of a Gaussian homogeneous environment where the covariance matrix $\mathbf{C}$ is perfectly known, the GLRT is:
\begin{equation}
\text{MIMO-MF}(\mathbf{x}) = \mathbf{x}^H \mathbf{C}^{-1} \mathbf{P} \left( \mathbf{P}^H \mathbf{C}^{-1} \mathbf{P} \right)^{-1} \mathbf{P}^H \mathbf{C}^{-1} \mathbf{x} \, .
\label{eq:MIMO_MF}
\end{equation}
This is the \emph{MIMO Optimum Gaussian Detector} (R-MIMO OGD) in \cite{Chong201004}, which is a multi-array generalization of the \emph{Matched Filter} test. One can note the very strong similarity with \eqref{eq:rao_dector}.
Its adaptive version is MIMO-AMF$_{SCM}$.
It seems useful to specify that this detector is relevant only in a Gaussian and perfectly homogeneous environment. Especially, exactly as in the single-array case \cite{Ollila2012}, the covariance estimator \eqref{eq:cov_Tyler} is defined up to a constant scalar. Detectors \eqref{eq:glrt_detector} and \eqref{eq:rao_dector} are invariant when the covariance is changed by a scale factor. This is a direct result of the partial homogeneity assumption. This is not the case for \eqref{eq:MIMO_MF}, thus the adaptive MIMO-AMF$_{TYL}$ version is not relevant and will not be used in the following.

These tests will be considered as performance bounds for the proposed detectors: M-NMF-G and M-NMF-R when $\mathbf{M}$ is known, M-ANMF-G$_{SCM}$, M-ANMF-R$_{SCM}$, M-ANMF-G$_{TYL}$ and M-ANMF-R$_{TYL}$ otherwise.

\subsection{Performance in Gaussian clutter}
\label{subsec:perf_Gaussian_clutter}

We have shown that all developed detectors are texture-CFAR. Unfortunately, the matrix CFAR property (distribution of the test remains identical whatever the true covariance matrix) is much more difficult to show. Therefore, we propose to perform a study on simulated data to check this matrix CFAR property. Figure \ref{fig:PFA_threshold_curves} experimentally demonstrates the CFAR behavior of the detectors concerning the covariance matrix in a Gaussian environment. On the left-hand side, we represent the false alarm probability of Rao's detector (for known and estimated $\mathbf{M}$) as a function of the detection threshold. On the right, these curves are plotted for the GLRT detector.

\begin{figure}[htbp]
	\centering
	\includegraphics[width = 1\linewidth]{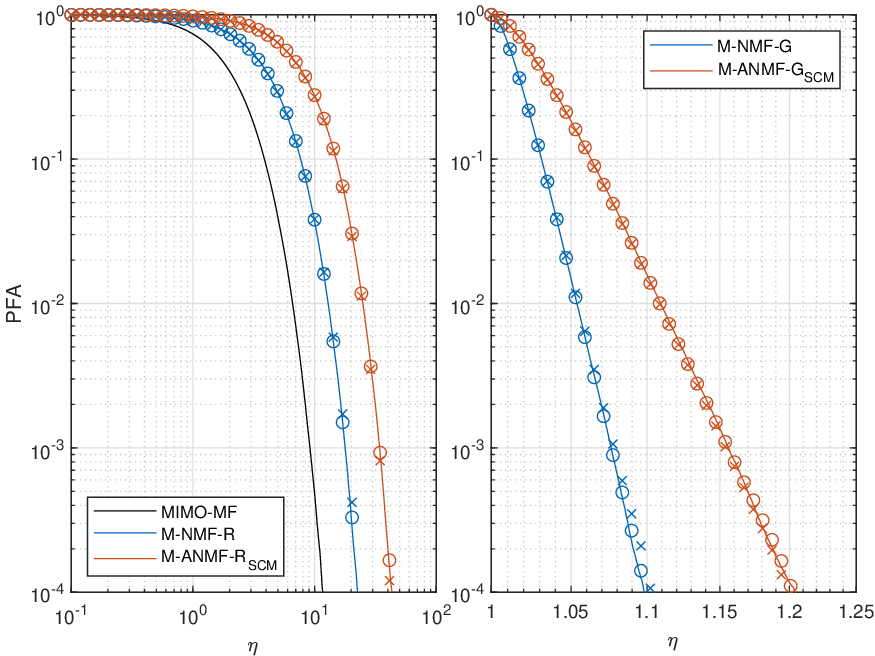}
	\caption{Pfa-threshold curves of the M-(A)NMF-R$_{SCM}$ and M-(A)NMF-G$_{SCM}$ detectors in Gaussian environment. Left: the Rao detector for known $\mathbf{M}$ (blue), for $\mathbf{M}$ estimated from $2 \times 2m$ secondary data (red), the OGD detector or \emph{Matched Filter} (black). Right: The GLRT detector for known $\mathbf{M}$ (blue), and estimated based on $2 \times 2m$ secondary data (red).}
	\label{fig:PFA_threshold_curves}
\end{figure}

The deviation of the adaptive detectors comes from the covariance estimation process: the red curve converges to the blue one when the number of secondary data used for the estimation of $\mathbf{M}$ increases. The markers represent curves for different covariances. In solid line $\beta = 3 \times 10^{-4}$, $\rho_1=0.4$, $\rho_2=0.9$. The cross markers are established from $\beta = 1$, $\rho_1=0.95$, $\rho_2=0.95$. Circular markers are for $\beta = 100$, $\rho_1=0.1$, $\rho_2=0.1$ and null anti-diagonal blocks. For very distinct covariances, the superposition of curves and markers underlines the \emph{matrix-CFAR} property of the detectors.

\begin{figure}[htbp]
	\centering
	\includegraphics[width = 0.45\linewidth]{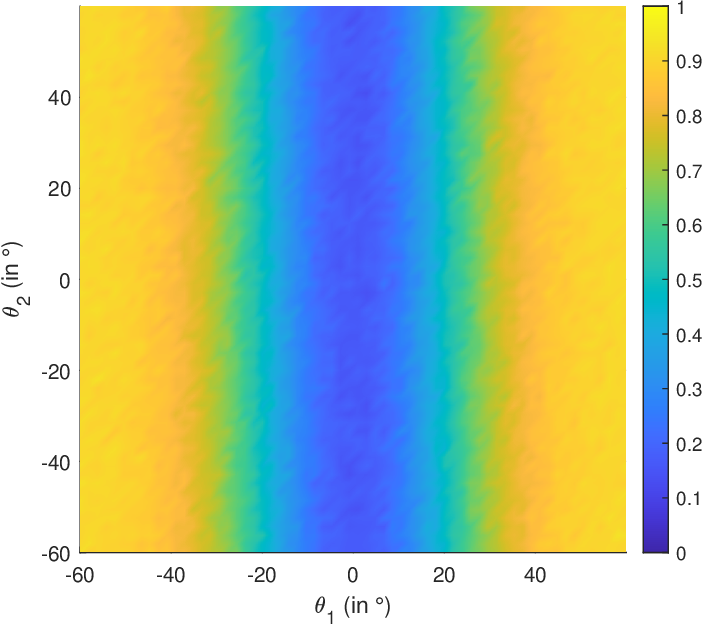}
	\includegraphics[width = 0.45\linewidth]{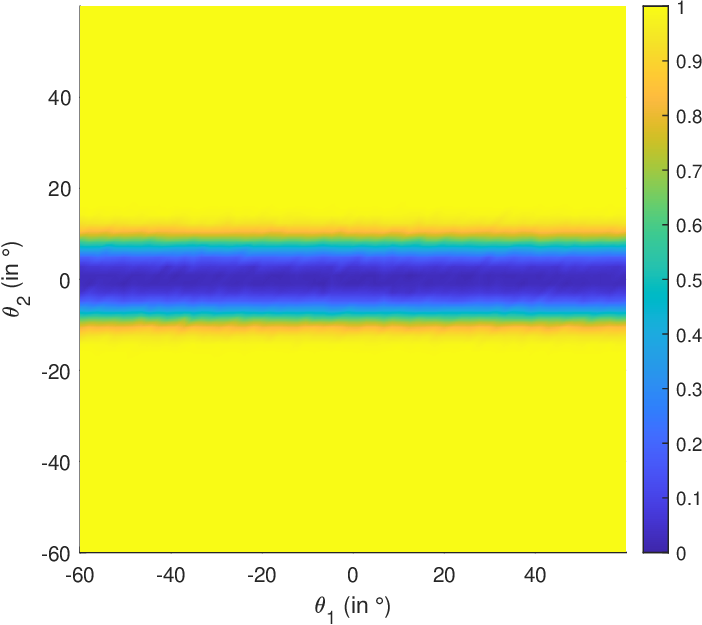}\\
 	\includegraphics[width = 0.45\linewidth]{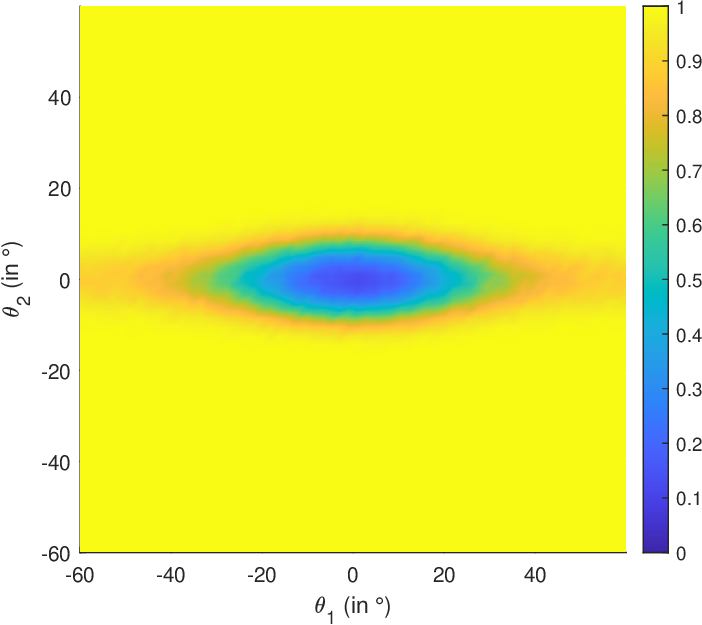}
	\caption{Probability of detection in Gaussian environment for known covariance matrix $\mathbf{M}$ as a function of $\theta_1$ and $\theta_2$ ($\beta = 3 \times 10^{-4}$, $\rho_1=0.4$, $\rho_2=0.9$, $SNR=-12~$dB and $PFA=10^{-2}$). Top left: NMF 1. Top right: NMF 2. Bottom left: M-NMF-R.}
	\label{fig:pd_theta_known_M}
\end{figure}

Figure \ref{fig:pd_theta_known_M} illustrates the value of merging arrays with a cross-shaped geometry. The detection probabilities of NMF 1, NMF 2, and M-NMF-R are plotted as a function of azimuth ($\theta_1$) and elevation ($\theta_2$) angles at fixed SNR. In the first angular dimension, which depends on the array considered, each ULA has a zone of least detectability close to 0\Deg. This is due to the correlated nature of clutter. The decrease in detection probabilities from +/-60\Deg\ to 0\Deg\ is a function of the correlation level ($\rho_1$ or $\rho_2$). In the second dimension, the detection performances are invariant. Thus, the poor detection performance near 0\Deg\ propagates in a whole direction of space: ``vertically'' for NMF 1 and ``horizontally'' for NMF 2. As an illustration, the probability of detection of NMF 1 in $\theta_1=0\Deg$ whatever $\theta_2$ is $PD=0.17$, and the probability of detection of NMF 2 in $\theta_2=0\Deg$ for all $\theta_1$ is 0.03. The M-NMF-R detector spatially minimizes this area of least detectability. In this case, the probability of detection at $\theta_1=0\Deg$ becomes greater than $0.8$ for $\lvert \theta_2 \rvert > 10.5\Deg$ and greater than $0.5$ for $\theta_2=0\Deg$ and $\lvert \theta_1 \rvert > 24\Deg$. The \emph{Rigde clutter} is minimized.

The probability of detection (PD) is plotted as a function of the signal-to-noise ratio (SNR) in Figure \ref{fig:pd_snr_known_M} when $\mathbf{M}$ is known. The MIMO-MF detector (shown in black), which requires a perfect knowledge of the covariance, is logically the most efficient. The single array NMF detectors (NMF 1 and 2, in purple and green) have comparable and the lowest performance. The M-NMF-I detector (in blue) between these curves assumes antenna independence. The proposed M-NMF-G (red) and M-NMF-R (yellow) detectors are both equivalent and superior to the M-NMF-I (0.2~dB at $PD=0.8$) and much more efficient than NMF 1 (2.5~dB) and NMF 2 tests (2~dB). The difference with MIMO-MF is slight, around 0.2~dB at $PD=0.8$. Detection performances by the Rao approach are slightly superior to those obtained by the GLRT.

\begin{figure}[htbp]
	\centering
    \begin{overpic}
        [width = 1\linewidth]{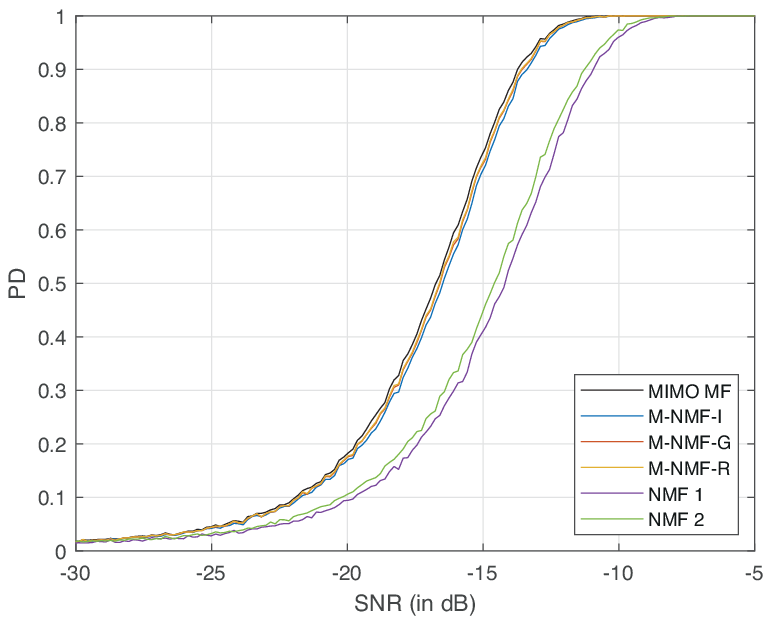}
        \put(12,43){\includegraphics[width = 0.44\linewidth]{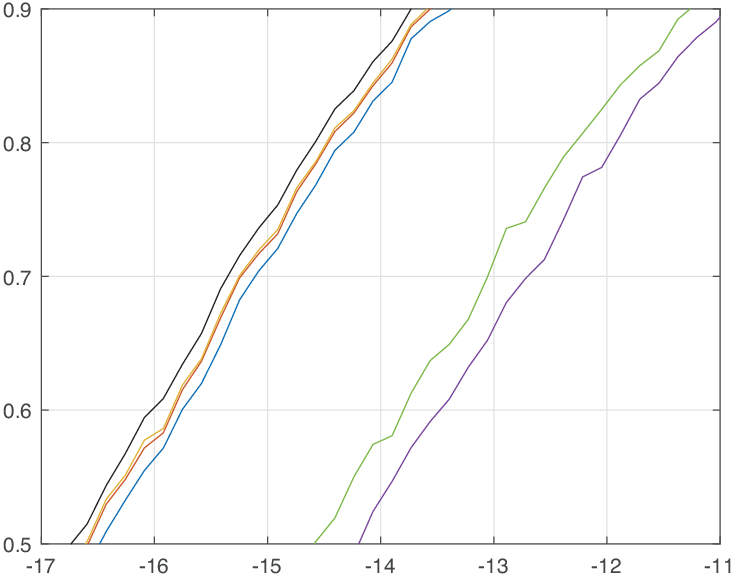}}
    \end{overpic}
	\caption{Probability of detection in Gaussian environment for known $\mathbf{M}$ ($\beta = 3 \times 10^{-4}$, $\rho_1=0.4$, $\rho_2=0.9$ and $PFA=10^{-2}$).}
	\label{fig:pd_snr_known_M}
\end{figure}

Figure \ref{fig:pd_snr_unknown_M} compares the performance of the tests in their adaptive versions. The MIMO-MF curve is shown in black dotted lines for illustrative purposes only (as it assumes the covariance known). It can be seen that when the SCM is evaluated on $2 \times 2m=256$ secondary data, the covariance estimation leads to a loss of 3~dB between the performances of the MIMO-MF test and its adaptive version MIMO-AMF, as expected from the Reed-Mallett-Brennan’s rule \cite{Reed1974}. In their adaptive versions, the proposed detectors offer equivalent performances to the MIMO-AMF detectors while offering additional robustness and flexibility conditions on a possible lack of knowledge of the covariance (estimated to be within two scale factors).
The gain compared to single array detectors, although reduced compared to the case where $\mathbf{M}$ is known, remains favorable and of the order of 0.5 to 1~dB at $PD=0.8$. In other words, for an SNR of -11~dB, $PD = 0.75$ for ANMF 1 and 0.85 for M-ANMF-G$_{SCM}$ (or M-ANMF-R$_{SCM}$).

\begin{figure}[htbp]
	\centering
    \begin{overpic}
        [width = 1\linewidth]{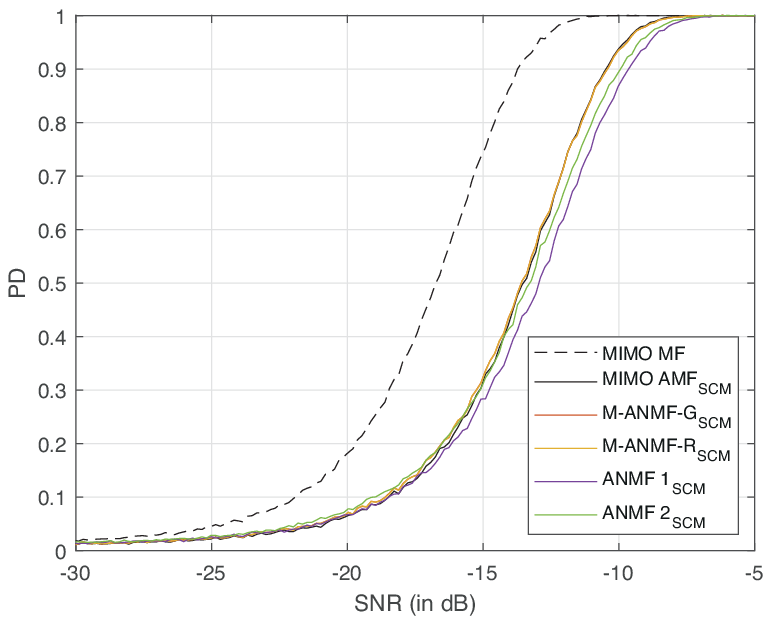}
        \put(12,43){\includegraphics[width = 0.44\linewidth]{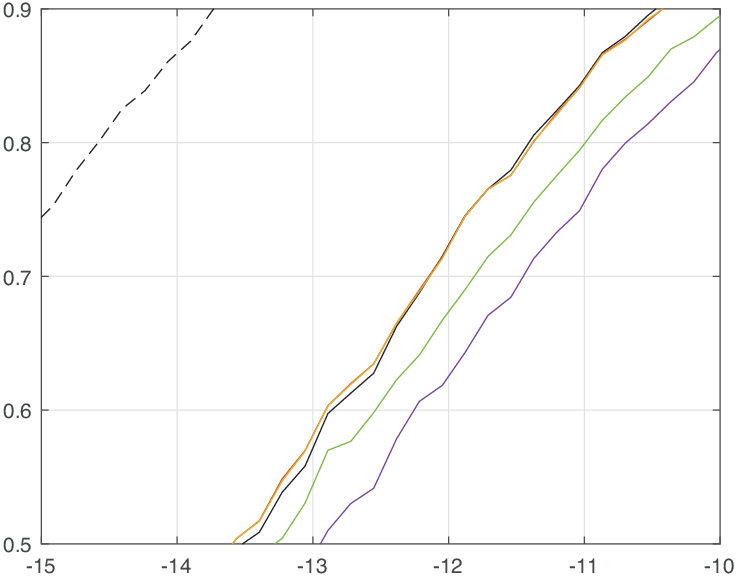}}
    \end{overpic}
	\caption{Probability of detection in Gaussian environment for unknown $\mathbf{M}$ ($PFA=10^{-2}$, SCMs are estimated based on $2 \times 2m$ secondary data).}
	\label{fig:pd_snr_unknown_M}
\end{figure}

\subsection{Performance in impulsive non-Gaussian clutter}
\label{subsec:pref_nonGaussian_clutter}

PFA-threshold curves in \emph{K}-distributed environment are displayed in Figure \ref{fig:PFA_threshold_curves_NG} to characterize the matrix-CFAR behavior. The detectors based on the new covariance matrix estimator systematically have lower detection thresholds than those obtained with the SCM. While optimal for a Gaussian clutter, the MIMO-MF detector is no longer suitable. The marker overlays highlight the matrix-CFAR behavior of the Rao detector in a non-Gaussian environment. The case of the GLRT detector is much less obvious: at this stage, it seems complicated to consider that the false alarm curves are strictly independent of the covariance matrix.

\begin{figure}[htbp]
	\centering
	\includegraphics[width = 1\linewidth]{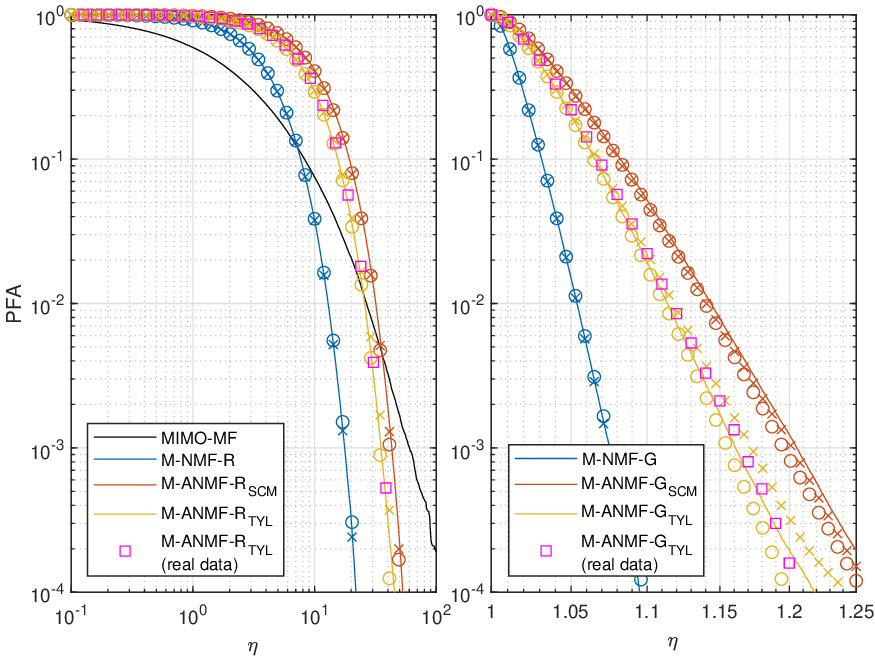}
	\caption{Pfa-threshold curves for Rao (left) and GLRT (right) detectors in non-Gaussian \emph{K}-distributed environment ($\nu = 0.5$). In blue $\mathbf{M}$ is known, in red $\mathbf{M}$ is estimated with the SCM, in yellow $\mathbf{M}$ is estimated with the new estimator $\widehat{\mathbf{M}}_{\text{2TYL}}$, the OGD or \emph{Matched Filter} detector is black. In magenta the curves are obtained on real data between 300~m and 315~m (see \ref{subsec:data_statistics}).}
	\label{fig:PFA_threshold_curves_NG}
\end{figure}

The detection performances are compared in Figure \ref{fig:pd_snr_unknown_M_NG}. For the Rao and GLRT detectors, using the estimator \eqref{eq:cov_Tyler} in an impulsive non-Gaussian environment induces a gain in detection performance of the order of 1.2~dB at $PD=0.8$ with the SCM. Compared to the best competitors the improvement is 3~dB at $PD=0.8$.

\begin{figure}[htbp]
	\centering
    \begin{overpic}
        [width = 1\linewidth]{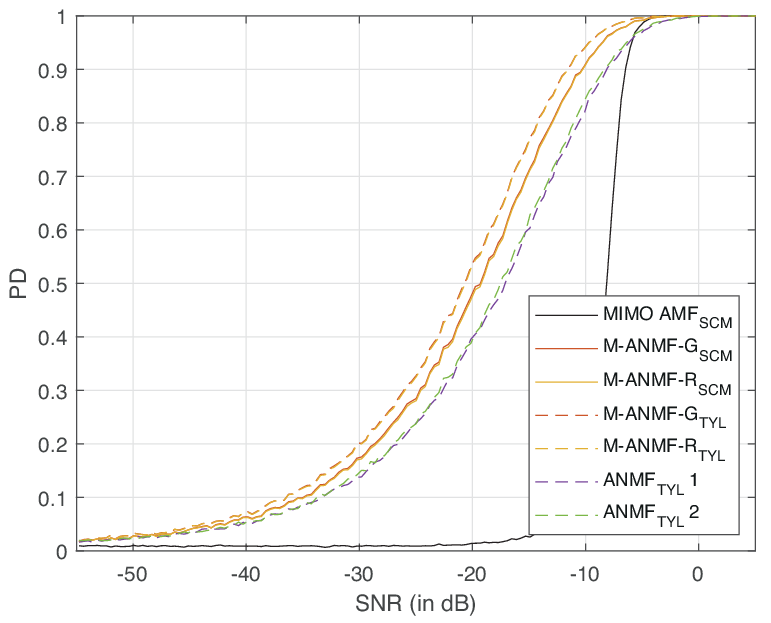}
        \put(12,43){\includegraphics[width = 0.44\linewidth]{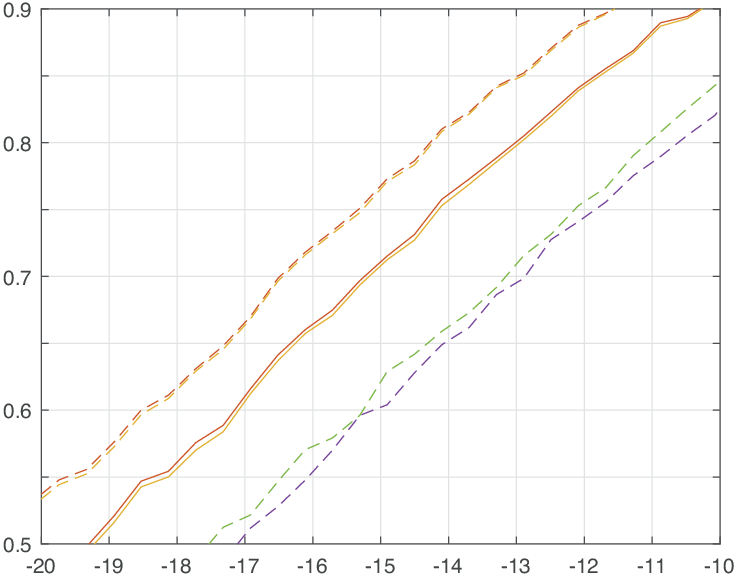}}
    \end{overpic}
	\caption{Probability of detection in \emph{K}-distributed environment for unknown $\mathbf{M}$ ($\nu = 0.5$, $PFA=10^{-2}$, $2 \times 2m$ secondary data). In solid lines, the detectors are built with the SCM estimate. In dotted-lines, the detectors are built with Tyler's (single array) or the new (dual arrays) estimator.}
	\label{fig:pd_snr_unknown_M_NG}
\end{figure}

\section{Experimental data with real target}
\label{sec:results_exp}

The sonar measurements were collected during an acquisition campaign in May 2022 in La Ciotat Bay, Mediterranean Sea. The experimental conditions are described in \ref{subsec:fls_experiment}, and the dataset of interest is shown in Figure \ref{fig:fls_swath}.
Echoes from the true target (right in Figure \ref{fig:drix_and_target}) are observed at the 2040th range cell.

\subsection{Overview of real data statistics}
\label{subsec:data_statistics}

1000 real data $2m$-samples are selected from the range bin 1000 to 1999 (i.e. from 78~m to 140~m in Figure \ref{fig:fls_swath}).
If the samples are complex Gaussian $\mathbf{x} \sim \mathcal{CN}(\mathbf{0}, \mathbf{M}) \in \mathbb{C}^{2m}$ the Hermitian form is Gamma distributed $\mathbf{x}^H \widehat{\mathbf{M}}^{-1} \mathbf{x} \sim \mathrm{Gam} \left( 2m , 1 \right)$. The modulus (or amplitude) of a standardized univariate marginal follows a Rayleigh distribution:
\begin{equation}
    p_x \left( x \right) = 2 x \exp \left( -x^2 \right) \, .
\end{equation}
Note that the PDF of the amplitude of a standardized univariate \emph{K}-distribution is given by \cite{Ollila2012}:
\begin{equation}
    p_x \left( x \right) = \cfrac{2 \sqrt{\nu}}{2^{\nu-1} \Gamma \left( \nu \right)} \left( 2 \sqrt{\nu} x \right)^{\nu} K_{\nu-1} \left( 2 \sqrt{\nu} x \right) \ ,
\end{equation}
where $\nu > 0$ is the shape parameter and $K_{l} \left( . \right)$ denotes the modified Bessel function of the second kind of order $l$.

Figure \ref{fig:real_data_distrib} (left) shows the quantile-quantile diagram of the quadratic form $\mathbf{x}^H \widehat{\mathbf{M}}^{-1} \mathbf{x}$ versus a Gamma distribution.
Quantile-quantile diagrams depict the scatter-plot of empirical quantiles versus the theoretical quantiles of a presupposed distribution. If the quantiles are similar, the plot looks like a straight line, and we can conclude that the empirical and theoretical distributions are identical.
In Figure \ref{fig:real_data_distrib} the curve deviates from a straight line as we move away from the central abscissa. The plot is not linear. The empirical and reference distributions are not the same and the real data is not complex Gaussian. On the right side of Figure \ref{fig:real_data_distrib} the empirical amplitude distribution of the first marginal is plotted as a histogram, Rayleigh's PDF is blue, the PDF of the amplitude of a standardized univariate \emph{K}-distribution is red. 
While the Rayleigh distribution gives a very poor fit, the agreement with a \emph{K}-distribution is very good. The analysis of these simple statistical indicators shows that empirical data are not distributed according to a Gaussian distribution and are rather close to a \emph{K}-distribution.

\begin{figure}[H]
	\begin{minipage}{0.47\linewidth}
	\includegraphics[width = 0.9\linewidth]{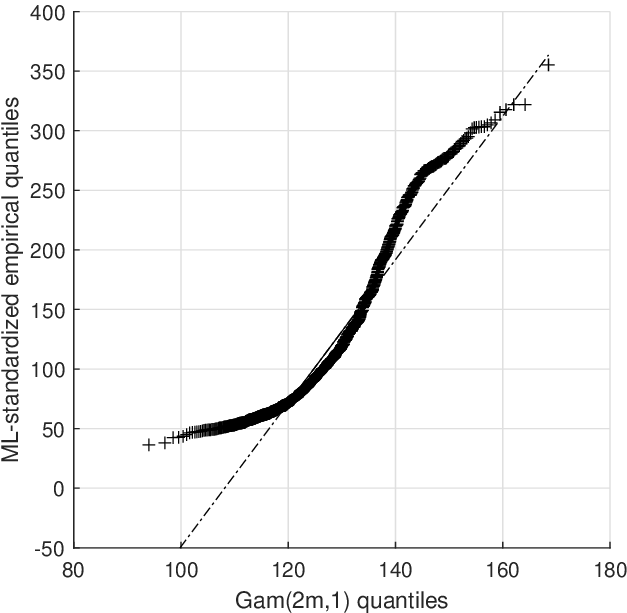}
	\end{minipage}
	\hfill
	\begin{minipage}{0.53\linewidth}
	\includegraphics[width = \linewidth]{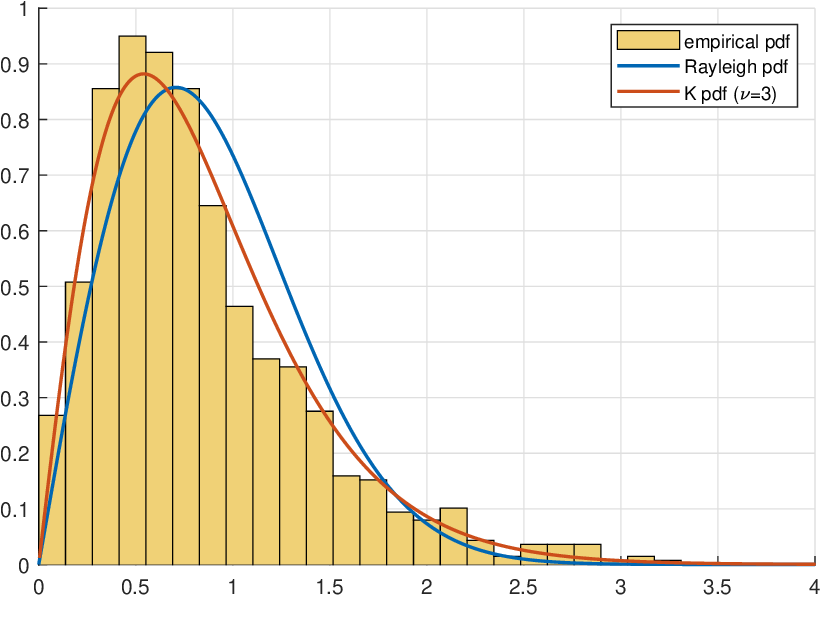}
	\end{minipage}
	\caption{Quantile-quantile diagram (left) and empirical distribution (right) on real data.}
     \label{fig:real_data_distrib}
\end{figure}

Figure \ref{fig:PFA_threshold_curves_NG} also shows the Pfa-threshold curves of the Rao and GLRT detectors obtained on real experimental data (from 300~m to 315~m on Figure \ref{fig:fls_swath}), with magenta markers. The markers appear to be superimposed on the simulated data. This good agreement validates our modeling of the environment and the ability to predefine the detection threshold. Note that since the SCM estimator is not texture-independent, the corresponding Pfa-threshold curves are not shown.

\subsection{Illustrations of detection test outputs}
\label{subsec:outputs}

Figure \ref{fig:anmf_output} provides examples of real data outputs from conventional ANMF detectors based on a single array. The real target is observed in range bin 2040 (143~m away from the sonar), at angular bin $\theta_1=26$ ($-12.4\Deg$ azimuth), and angular bin $\theta_2=37$ ($8.6\Deg$ elevation).

\begin{figure}[htbp]
	\centering
	\includegraphics[width = 0.49\linewidth]{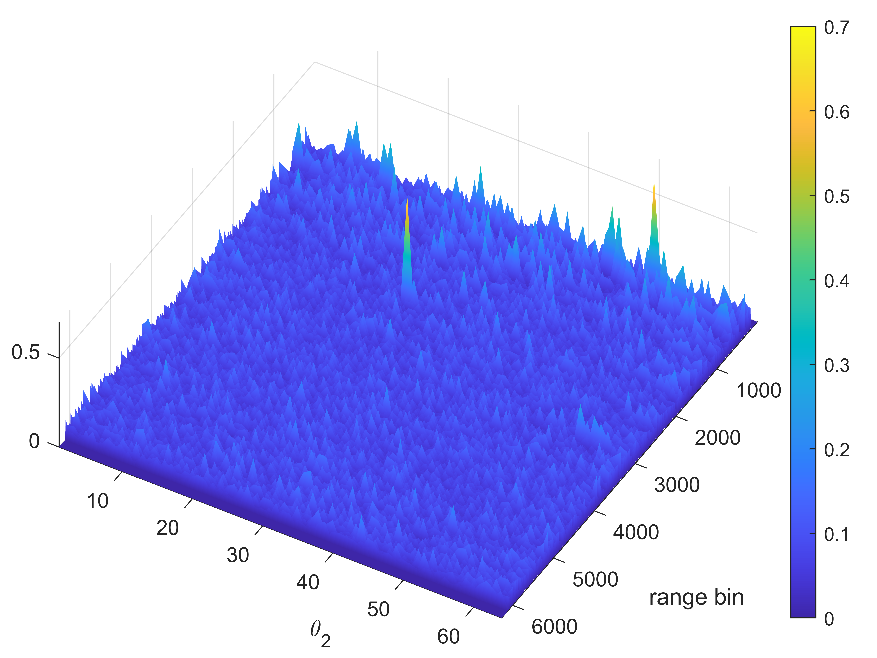}
	\includegraphics[width = 0.49\linewidth]{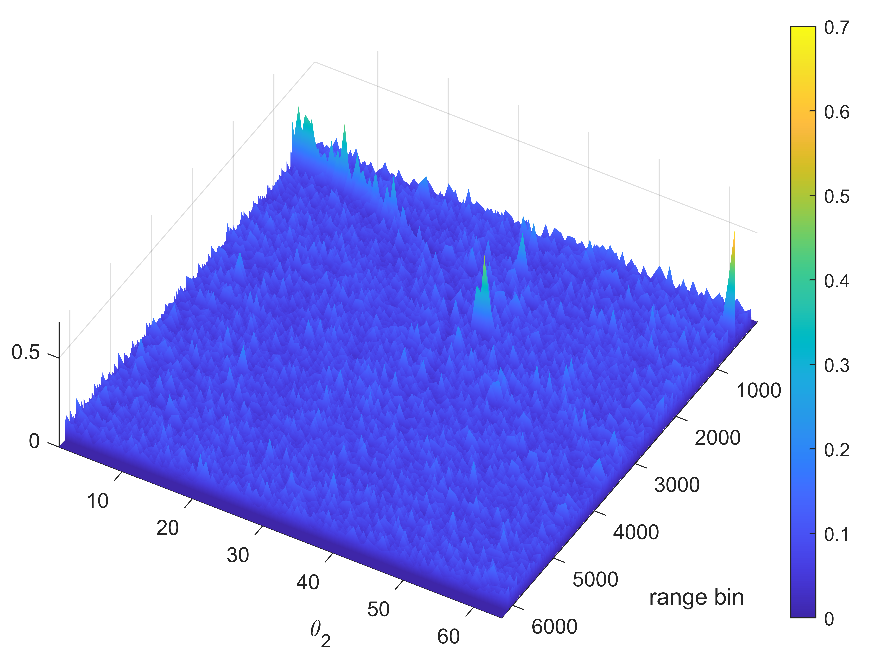}
	\caption{ANMF$_{SCM}$ detector outputs on real sonar data. A target is observed on array 1 (left) at coordinates $(26, \, 2040)$ and on array 2 (right) at coordinates $(37, \, 2040)$. The SCM is built from 256 secondary data.}
	\label{fig:anmf_output}
\end{figure}

Figure \ref{fig:mimo_output} shows the outputs of the Rao and GLRT detectors applied simultaneously to both arrays at the specific range bin 2040. These subfigures are not directly comparable with each other or with Figure \ref{fig:anmf_output}. The target is perfectly located in azimuth and elevation.

\begin{figure}[htbp]
	\centering
	\includegraphics[width = 0.49\linewidth]{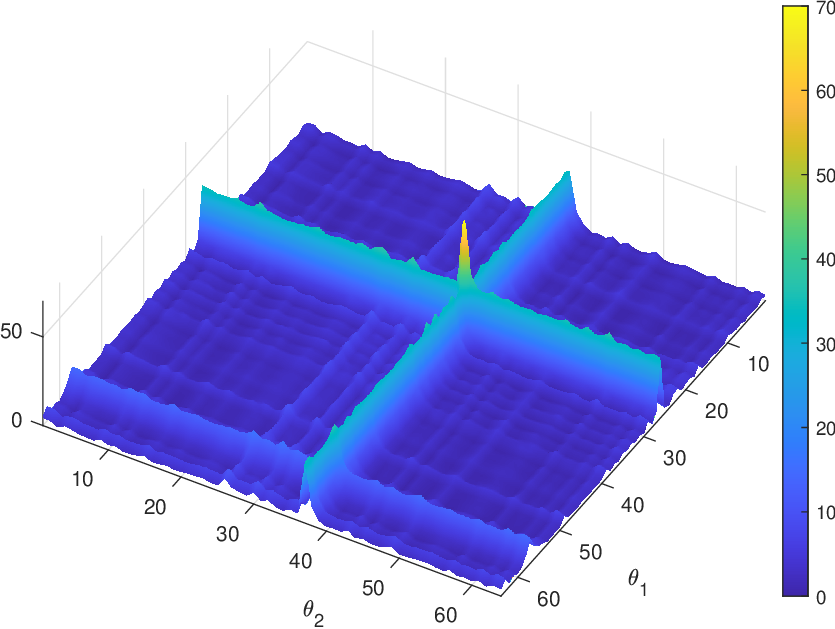}
	\includegraphics[width = 0.49\linewidth]{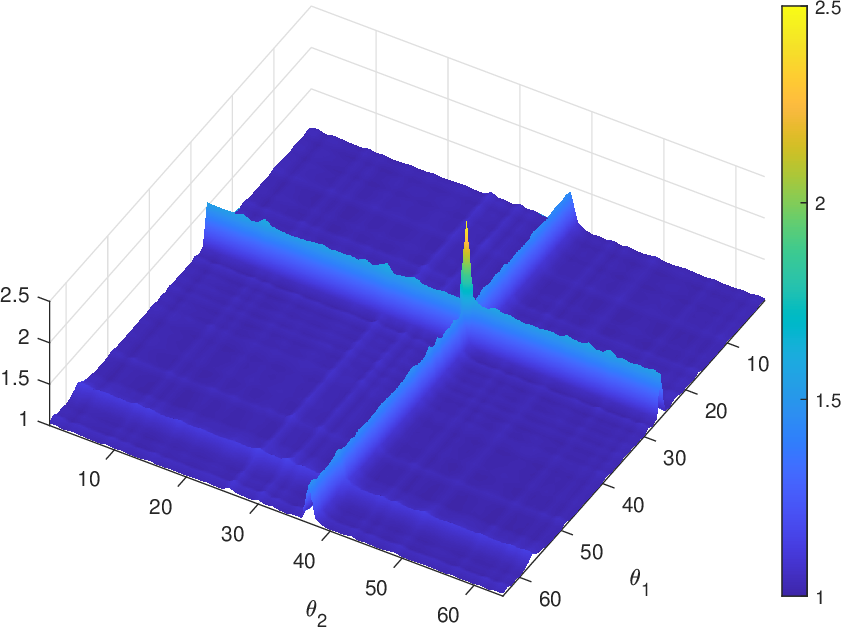}
	\caption{Outputs of the M-ANMF-R$_{SCM}$ (left) and M-ANMF-G$_{SCM}$ (right) detectors. The real target is at coordinates (37,\, 26). The SCM is built from 256 secondary data.}
	\label{fig:mimo_output}
\end{figure}

\subsection{Robustness to training data corruption}
\label{subsec:trainingdata}

Previously, we assumed the availability of a set of $K \geq 2m$ i.i.d. secondary data, supposed free of signal components and sharing statistic properties with the noise in the cell under test. In practice, such data can be obtained by processing samples in spatial proximity to the range bin being evaluated, and the absence of signals is not always verified or checkable. In particular, this assumption is no longer valid if another target is present in these secondary data.

Figure \ref{fig:corrupted_rao} illustrates the robustness of the covariance matrix estimators to data corruption. A synthetic target is added 100 samples away from the real target and contaminates the secondary dataset. On the left-hand side, the output of M-ANMF-R$_{SCM}$ is strongly degraded. The SCM is not robust to the presence of outliers. The target is hardly discernible, and the maximum is wrongly located. Under the same conditions, the response of the M-ANMF-R$_{TYL}$ detector is visualized in the right part. Although degraded compared to Figure \ref{fig:mimo_output}, the behavior is still largely usable. The target presence is easily identified, and the new estimator \eqref{eq:cov_Tyler} is much more resistant to data contamination.

\begin{figure}[htbp]
	\centering
	\includegraphics[width = 0.49\linewidth]{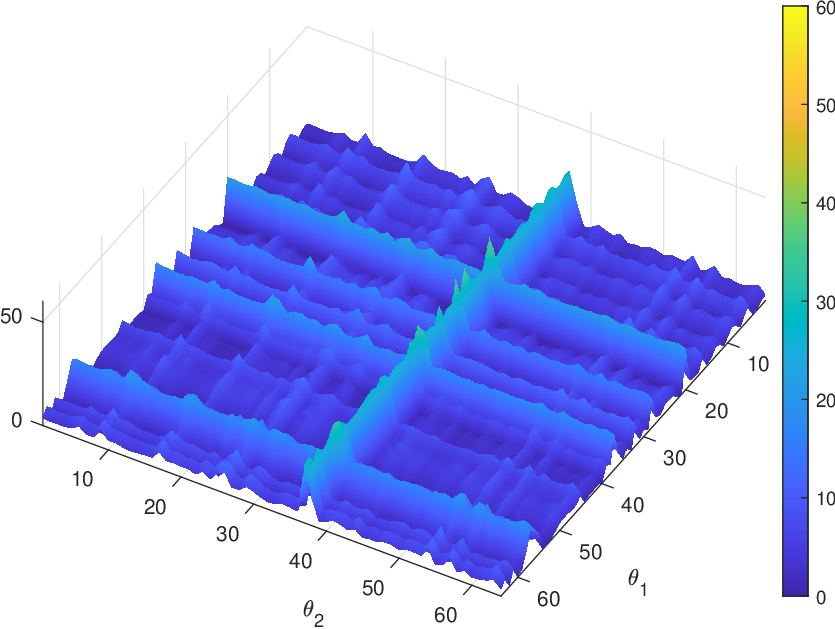}
	\includegraphics[width = 0.49\linewidth]{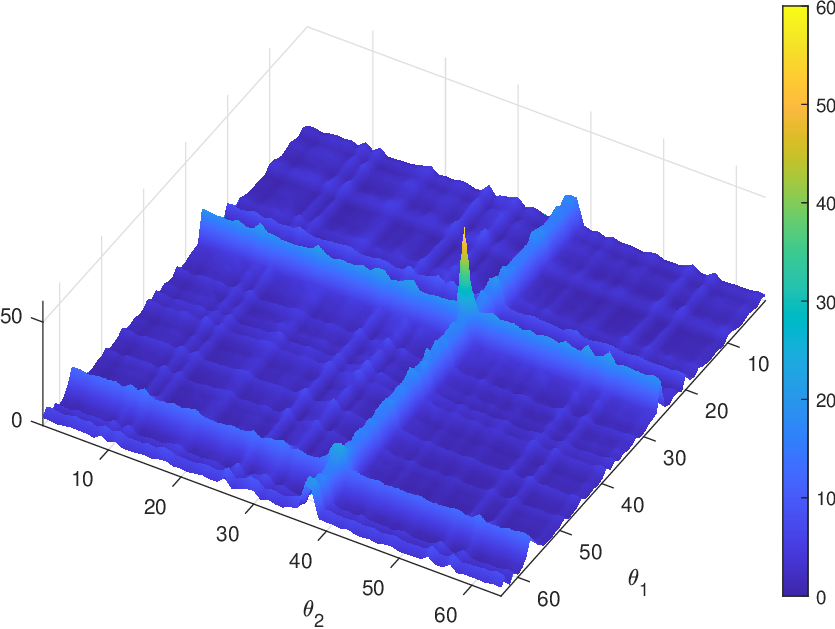}
	\caption{Rao detector outputs based on the SCM (left) or Tyler (right) with secondary data corruption. The covariance matrix estimators are based on 256 secondary data.}
	\label{fig:corrupted_rao}
\end{figure}

\section{Conclusions}
\label{sec:conclusions}

In this paper, we considered the problem of adaptive point target detection by a correlated multi-arrays Mills Cross sonar system. Using a 2-step approach, we first derived two new detectors that are robust to differences in target amplitudes and to unknown scaling factors on the covariances. Subsequently, we have introduced an innovative estimator of the covariance matrix suitable to any non-Gaussian MSG environment. By these very general assumptions, the framework of the study can therefore concern systems with co-located or remote arrays.

Experimental results show that the detection performance is up to 3~dB better than conventional approaches. The detectors cover a larger detection area and are particularly robust to spikes, impulsive noise, and data contamination.

Future work will focus on establishing a theoretical demonstration of the matrix-CFAR behavior of these detectors, and on generalizing solutions for different numbers and geometries of arrays.

\section*{Acknowledgments}
\label{sec:acknowledgement}

This work has been done thanks to the facilities offered by the Univ. Savoie Mont Blanc - CNRS/IN2P3 MUST computing center.

\appendices
\section{GLRT's derivation}
\label{sec:appendixA}

For the following, and for ease of reading, the punctuation mark ``tilde" will be omitted. Thus $\widetilde{\boldsymbol{\Sigma}}$ and $\tilde{\sigma}_i$ will simply be denoted as $\boldsymbol{\Sigma}$ and $\sigma_i$. The same is true for their respective estimates.

\subsection{Maximum Likelihood Estimator of $\boldsymbol{\Sigma}$ under $H_0$}
\label{subsec:Sigma0_derivation}

The MLE $\hat{\boldsymbol{\Sigma}}_0$ is derived from the log-likelihood function:
\begin{equation*}
\ln p_{\mathbf{x}}(\mathbf{x}; \boldsymbol{\Sigma}, H_0) =  -mL \ln \pi +2 \ln \lvert \boldsymbol{\Sigma}^{-1} \rvert - \ln \lvert \mathbf{M} \rvert - \left( {\mathbf{x}}^H \boldsymbol{\Sigma}^{-1} \mathbf{M}^{-1} \boldsymbol{\Sigma}^{-1} \mathbf{x} \right) \, ,
\end{equation*} whose derivative relative to $\boldsymbol{\Sigma}^{-1}$ is:
\begin{equation*}
\cfrac{\partial \ln p_{\mathbf{x}}(\mathbf{x}; \boldsymbol{\Sigma}, H_0)}{\partial \boldsymbol{\Sigma}^{-1}} = 2 \cfrac{\partial \ln \lvert \boldsymbol{\Sigma}^{-1} \rvert}{\partial \boldsymbol{\Sigma}^{-1}} - \cfrac{\partial \mathbf{x} ^H \boldsymbol{\Sigma}^{-1} \mathbf{M}^{-1} \boldsymbol{\Sigma}^{-1} \mathbf{x}}{\partial \boldsymbol{\Sigma}^{-1}} \, .
\end{equation*}
Knowing that (\cite{Petersen2012} (57))
$\cfrac{\partial \ln \lvert \boldsymbol{\Sigma}^{-1} \rvert}{\partial \boldsymbol{\Sigma}^{-1}} = \boldsymbol{\Sigma}$, and (\cite{Petersen2012} (82)) $\cfrac{\partial \mathbf{x}^H \boldsymbol{\Sigma}^{-1} \mathbf{M}^{-1} \boldsymbol{\Sigma}^{-1} \mathbf{x}}{\partial \boldsymbol{\Sigma}^{-1}}  = 2 \text{Re} \left( \mathbf{M}^{-1} \boldsymbol{\Sigma}^{-1} \mathbf{x} \, \mathbf{x}^H \right)$ we have:
\begin{equation*}
\cfrac{\partial \ln p_{\mathbf{x}}(\mathbf{x}; \boldsymbol{\Sigma}, H_0)}{\partial \boldsymbol{\Sigma}^{-1}} = 2 \boldsymbol{\Sigma} - 2 \text{Re} \left( \mathbf{M}^{-1} \boldsymbol{\Sigma}^{-1} \mathbf{x} \, \mathbf{x}^H \right) \, ,
\end{equation*} which leads to: 
\begin{equation}
\hat{\boldsymbol{\Sigma}}_0 = \text{Re} \left( \mathbf{M}^{-1} \hat{\boldsymbol{\Sigma}}_0^{-1} \mathbf{x} \, \mathbf{x}^H \right)\,.
\end{equation}
Expanding this matrix product with $\hat{\boldsymbol{\Sigma}}_0 = \begin{bmatrix} 
\hat{\sigma}_{1_0} \mathbf{I}_m & \mathbf{0} \\ 
\mathbf{0} & \hat{\sigma}_{2_0} \mathbf{I}_m \\
\end{bmatrix}$, we have:
\begin{equation*}
\hat{\sigma}_{1_0} \, \mathbf{I}_m = \text{Re} \left( \mathbf{M}_{11}^{-1}\cfrac{\mathbf{x}_1 \, \mathbf{x}_1^H}{\hat{\sigma}_{1_0}} + \mathbf{M}_{12}^{-1}\cfrac{\mathbf{x}_2 \, \mathbf{x}_1^H}{\hat{\sigma}_{2_0}} \right) \, ,
\end{equation*} 
where $\mathbf{M}_{11}^{-1}$ and $\mathbf{M}_{12}^{-1}$ are the 11th and 12th block of $\mathbf{M}^{-1}$ respectively. Using the trace operator:
\begin{align*}
	m \, \hat{\sigma}_{1_0} & = \tr \left[ \text{Re} \left( \mathbf{M}_{11}^{-1}\cfrac{\mathbf{x}_1 \, \mathbf{x}_1^H}{\hat{\sigma}_{1_0}} + \mathbf{M}_{12}^{-1}\cfrac{\mathbf{x}_2 \, \mathbf{x}_1^H}{\hat{\sigma}_{2_0}} \right) \right] \, ,\\	
	& = \cfrac{\mathbf{x}_1^H \mathbf{M}_{11}^{-1} \mathbf{x}_1}{\hat{\sigma}_{1_0}} + \text{Re} \left( \cfrac{\mathbf{x}_1^H \mathbf{M}_{12}^{-1} \mathbf{x}_2}{\hat{\sigma}_{2_0}} \right)\, .
\end{align*}
We then have:
\begin{equation}
\hat{\sigma}_{1_0}  = 
\cfrac{1}{\hat{\sigma}_{1_0}} \cfrac{\mathbf{x}_1^H \mathbf{M}_{11}^{-1} \mathbf{x}_1}{m} + 
\cfrac{1}{\hat{\sigma}_{2_0}} \cfrac{\text{Re} \left( \mathbf{x}_1^H \mathbf{M}_{12}^{-1} \mathbf{x}_2 \right)}{m}\, ,
\label{eq:sigma01}
\end{equation} and
\begin{equation}
\hat{\sigma}_{2_0}  = 
\cfrac{1}{\hat{\sigma}_{1_0}} \cfrac{\text{Re} \left( \mathbf{x}_2^H \mathbf{M}_{21}^{-1} \mathbf{x}_1 \right)}{m} +
\cfrac{1}{\hat{\sigma}_{2_0}} \cfrac{\mathbf{x}_2^H \mathbf{M}_{22}^{-1} \mathbf{x}_2}{m} \, .
\label{eq:sigma02}
\end{equation}

Denoting $a_1=\cfrac{\mathbf{x}_1^H \mathbf{M}_{11}^{-1} \mathbf{x}_1}{m}$, $a_{12}=\cfrac{\text{Re} \left( \mathbf{x}_1^H \mathbf{M}_{12}^{-1} \mathbf{x}_2 \right)}{m}$, \\
$a_2=\cfrac{\mathbf{x}_2^H \mathbf{M}_{22}^{-1} \mathbf{x}_2}{m}$ and using \eqref{eq:sigma01} and \eqref{eq:sigma02}:
\begin{equation}
\left\lbrace
\begin{array}{ll}
\hat{\sigma}_{1_0}^2 = a_1 + \cfrac{\hat{\sigma}_{1_0}}{\hat{\sigma}_{2_0}} \, a_{12} \\
\hat{\sigma}_{2_0}^2 = \cfrac{\hat{\sigma}_{2_0}}{\hat{\sigma}_{1_0}} \, a_{12} + a_2\, ,
\end{array}
\right.\
\label{eq:syst_L_2}
\end{equation}
or:
\begin{equation*}
\left\lbrace
\begin{array}{ll}
\hat{\sigma}_{1_0}^2 = a_1 + \cfrac{\hat{\sigma}_{1_0}}{\hat{\sigma}_{2_0}} \, a_{12} \\
\hat{\sigma}_{1_0}^2  = \cfrac{\hat{\sigma}_{1_0}}{\hat{\sigma}_{2_0}} \, a_{12} + \cfrac{\hat{\sigma}_{1_0}^2}{\hat{\sigma}_{2_0}^2} \, a_2\, .
\end{array}
\right.\
\end{equation*}
By equalization of right-hand terms:
\begin{equation*}
a_1 = \cfrac{\hat{\sigma}_{1_0}^2}{\hat{\sigma}_{2_0}^2} \, a_2\, , 
\end{equation*}
and keeping the positive solution: $\cfrac{\hat{\sigma}_{1_0}}{\hat{\sigma}_{2_0}} = \sqrt{\cfrac{a_1}{a_2}} \, $, we obtain from \eqref{eq:syst_L_2}:
\begin{equation}
\left\lbrace
	\begin{array}{ll}
	\hat{\sigma}_{1_0}^2 = a_1 + \sqrt{\cfrac{a_1}{a_2}} \, a_{12} \\
	\hat{\sigma}_{2_0}^2 = a_2 + \sqrt{\cfrac{a_2}{a_1}} \, a_{12} \, ,
	\end{array}
\right.\
\end{equation}
and
\begin{equation}
\hat{\sigma}_{1_0}^2 \, \hat{\sigma}_{2_0}^2 = \left( \sqrt{a_1 \, a_2} + a_{12} \right) ^2 .
\end{equation}

\subsection{Maximum Likelihood Estimator of $\boldsymbol{\alpha}$ under $H_1$}
\label{subsec:Alpha_derivation}

The ML estimates $\hat{\boldsymbol{\alpha}} = \left[ \hat{\alpha}_1 \, \hat{\alpha}_2 \right] ^T$ is found minimizing $\left( \mathbf{x} - \mathbf{P} \boldsymbol{\alpha} \right) ^H \mathbf{C}^{-1} \left( \mathbf{x} - \mathbf{P} \boldsymbol{\alpha} \right)$ with respect to $\boldsymbol{\alpha}$ as (\cite{Kay1993} (15.50)):
\begin{equation}
	\begin{array}{ll}
	\hat{\boldsymbol{\alpha}} \!\!\!
	& =  \left( \mathbf{P}^{H} \mathbf{C}^{-1} \mathbf{P} \right)^{-1} \mathbf{P}^{H} \mathbf{C}^{-1} \mathbf{x}\, , \\
	 \!\!\!& =  \left( \mathbf{P}^{H} \boldsymbol{\Sigma}^{-1} \mathbf{M}^{-1} \boldsymbol{\Sigma}^{-1} \mathbf{P} \right)^{-1} \mathbf{P}^{H} \boldsymbol{\Sigma}^{-1} \mathbf{M}^{-1} \boldsymbol{\Sigma}^{-1} \mathbf{x} \, .
	\end{array}
	\label{eq:alpha}
\end{equation}

\subsection{Maximum Likelihood Estimator of $\boldsymbol{\Sigma}$ under $H_1$}
\label{subsec:Sigma1_derivation}

The derivative of the log-likelihood function under $H_1$ with respect to $\boldsymbol{\Sigma}^{-1}$ is:
\begin{equation*}
\cfrac{\partial \ln p_{\mathbf{x}}(\mathbf{x}; \hat{\boldsymbol{\alpha}}, \boldsymbol{\Sigma}, H_1)}{\partial \boldsymbol{\Sigma}^{-1}} = 2 \, \cfrac{\partial \ln \lvert \boldsymbol{\Sigma}^{-1} \rvert}{\partial \boldsymbol{\Sigma}^{-1}}
 - \cfrac{\partial \left( \mathbf{x} - \mathbf{P} \hat{\boldsymbol{\alpha}} \right) ^H \boldsymbol{\Sigma}^{-1} \mathbf{M}^{-1} \boldsymbol{\Sigma}^{-1} \left( \mathbf{x} - \mathbf{P} \hat{\boldsymbol{\alpha}} \right)}{\partial \boldsymbol{\Sigma}^{-1}} \, .
\end{equation*}
Furthermore:
\begin{multline*}
\left( \mathbf{x} - \mathbf{P} \hat{\boldsymbol{\alpha}} \right) ^H \boldsymbol{\Sigma}^{-1} \mathbf{M}^{-1} \boldsymbol{\Sigma}^{-1} \left( \mathbf{x} - \mathbf{P} \hat{\boldsymbol{\alpha}} \right) \\
= \mathbf{x}^H \boldsymbol{\Sigma}^{-1} \mathbf{M}^{-1} \boldsymbol{\Sigma}^{-1} \mathbf{x} - \hat{\boldsymbol{\alpha}}^H \mathbf{P}^H \boldsymbol{\Sigma}^{-1} \mathbf{M}^{-1} \boldsymbol{\Sigma}^{-1} \mathbf{x} \\
= \mathbf{x}^H \boldsymbol{\Sigma}^{-1} \mathbf{M}^{-1} \boldsymbol{\Sigma}^{-1} \mathbf{x} - \mathbf{x}^H \boldsymbol{\Sigma}^{-1} \mathbf{M}^{-1} \boldsymbol{\Sigma}^{-1} \mathbf{P} 
\left( \mathbf{P}^H \boldsymbol{\Sigma}^{-1} \mathbf{M}^{-1} \boldsymbol{\Sigma}^{-1} \mathbf{P} \right) ^{-1} \\ \mathbf{P}^H \boldsymbol{\Sigma}^{-1} \mathbf{M}^{-1} \boldsymbol{\Sigma}^{-1} \mathbf{x}\, .
\end{multline*}
If 
$\mathbf{Z}=\begin{bmatrix}
\sigma_1 & 0 \\ 
0 & \sigma_2 \\
\end{bmatrix}$, we then notice that 
$\boldsymbol{\Sigma}^{-1} \mathbf{P} = \mathbf{P} \mathbf{Z}^{-1}$ and $\mathbf{P}^H \boldsymbol{\Sigma}^{-1} = \mathbf{Z}^{-1} \mathbf{P}^H$. Thus:
\begin{align*}
\mathbf{x}^H \boldsymbol{\Sigma}^{-1} \mathbf{M}^{-1} \boldsymbol{\Sigma}^{-1} \mathbf{P} \left( \mathbf{P}^H \boldsymbol{\Sigma}^{-1} \mathbf{M}^{-1} \boldsymbol{\Sigma}^{-1} \mathbf{P} \right) ^{-1} 
 \mathbf{P}^H \boldsymbol{\Sigma}^{-1} \mathbf{M}^{-1} \boldsymbol{\Sigma}^{-1} \mathbf{x} \\
= \mathbf{x}^H \boldsymbol{\Sigma}^{-1} \mathbf{M}^{-1} \mathbf{P} \mathbf{Z}^{-1} \left( \mathbf{Z}^{-1} \mathbf{P}^H \mathbf{M}^{-1} \mathbf{P} \mathbf{Z}^{-1} \right) ^{-1} 
 \mathbf{Z}^{-1} \mathbf{P}^H \mathbf{M}^{-1} \boldsymbol{\Sigma}^{-1} \mathbf{x} \\
= \mathbf{x}^H \boldsymbol{\Sigma}^{-1} \mathbf{M}^{-1} \mathbf{P} \left( \mathbf{P}^H \mathbf{M}^{-1} \mathbf{P} \right) ^{-1} \mathbf{P}^H \mathbf{M}^{-1} \boldsymbol{\Sigma}^{-1} \mathbf{x}\, .
\end{align*}

By denoting $\mathbf{D}^{-1} = \mathbf{M}^{-1} \mathbf{P} \left( \mathbf{P}^H \mathbf{M}^{-1} \mathbf{P} \right) ^{-1} \mathbf{P}^H \mathbf{M}^{-1}$:
\begin{equation*}
\left( \mathbf{x} - \mathbf{P} \hat{\boldsymbol{\alpha}} \right) ^H \boldsymbol{\Sigma}^{-1} \mathbf{M}^{-1} \boldsymbol{\Sigma}^{-1} \left( \mathbf{x} - \mathbf{P} \hat{\boldsymbol{\alpha}} \right) 
= \mathbf{x}^H \boldsymbol{\Sigma}^{-1} \mathbf{M}^{-1} \boldsymbol{\Sigma}^{-1} \mathbf{x} - \mathbf{x}^H \boldsymbol{\Sigma}^{-1} \mathbf{D}^{-1} \boldsymbol{\Sigma}^{-1} \mathbf{x}\, .
\end{equation*}
So:
\begin{equation*}
 \cfrac{\partial \left( \mathbf{x} - \mathbf{P} \hat{\boldsymbol{\alpha}} \right) ^H \boldsymbol{\Sigma}^{-1} \mathbf{M}^{-1} \boldsymbol{\Sigma}^{-1} \left( \mathbf{x} - \mathbf{P} \hat{\boldsymbol{\alpha}} \right)}{\partial \boldsymbol{\Sigma}^{-1}} \, = 2 \text{Re} \left[ \left( \mathbf{M}^{-1} - \mathbf{D}^{-1} \right) \boldsymbol{\Sigma}^{-1} \mathbf{x} \, \mathbf{x}^H \right]\, .
\end{equation*}
Finally:
\begin{equation*}
\cfrac{\partial \ln p_{\mathbf{x}}(\mathbf{x}; \hat{\boldsymbol{\alpha}}, \boldsymbol{\Sigma}, H_1)}{\partial \boldsymbol{\Sigma}^{-1}} = 2 \boldsymbol{\Sigma} - 2 \text{Re} \left[ \left( \mathbf{M}^{-1} - \mathbf{D}^{-1} \right) \boldsymbol{\Sigma}^{-1} \mathbf{x} \, \mathbf{x}^H \right] \, .
\end{equation*}
The minimum is given by:
\begin{equation}
\hat{\boldsymbol{\Sigma}}_1 = \text{Re} \left[ \left( \mathbf{M}^{-1} - \mathbf{D}^{-1} \right) \hat{\boldsymbol{\Sigma}}_1^{-1} \mathbf{x} \, \mathbf{x}^H \right] \, .
\label{eq:sigma_1}
\end{equation}

The rest is identical to paragraph \ref{subsec:Sigma0_derivation}.

\subsection{Expression of the GLRT}
\label{subsec:GLRT_expression}

As $\mathbf{x}^H \hat{\boldsymbol{\Sigma}}_0^{-1} \mathbf{M}^{-1} \hat{\boldsymbol{\Sigma}}_0^{-1} \mathbf{x}$ is a real positive scalar, we have:
\begin{align*}
	\mathbf{x}^H \hat{\boldsymbol{\Sigma}}_0^{-1} \mathbf{M}^{-1} \hat{\boldsymbol{\Sigma}}_0^{-1} \mathbf{x}
	& = \text{Re} \left[ \tr \left( \hat{\boldsymbol{\Sigma}}_0^{-1} \mathbf{M}^{-1} \hat{\boldsymbol{\Sigma}}_0^{-1} \mathbf{x} \, \mathbf{x}^H  \right) \right] \, , \\
	& = \tr \left[ \hat{\boldsymbol{\Sigma}}_0^{-1} \text{Re} \left( \mathbf{M}^{-1} \hat{\boldsymbol{\Sigma}}_0^{-1} \mathbf{x} \, \mathbf{x}^H  \right) \right] \, , \\
	& = m L \, .
\end{align*}
So the two PDF take the form:
\begin{eqnarray*}
	p_{\mathbf{x}}(\mathbf{x}; \hat{\boldsymbol{\Sigma}}_0, H_0) \!\!\! & = & \!\!\! \cfrac{1}{\pi^{mL} \lvert \hat{\boldsymbol{\Sigma}}_0 \rvert ^2 \lvert \mathbf{M}  \rvert} \exp \left( -mL \right)\, ,\\
	p_{\mathbf{x}}(\mathbf{x}; \hat{\boldsymbol{\alpha}}, \hat{\boldsymbol{\Sigma}}_1, H_1) \!\!\! & = & \!\!\! \cfrac{1}{\pi^{mL} \lvert \hat{\boldsymbol{\Sigma}}_1 \rvert ^2 \lvert \mathbf{M}  \rvert} \exp \left( -mL \right) \, .
\end{eqnarray*}
The GLRT is expressed as:
\begin{equation}
	L_G(\mathbf{x}) = \cfrac{\cfrac{1}{\lvert \hat{\boldsymbol{\Sigma}}_1 \rvert ^2}}{\cfrac{1}{\lvert \hat{\boldsymbol{\Sigma}}_0 \rvert ^2}} =
	\cfrac{\lvert \hat{\boldsymbol{\Sigma}}_0 \rvert ^2}{\lvert \hat{\boldsymbol{\Sigma}}_1 \rvert ^2}\, ,
\end{equation}
which can be thought of as a generalized variance ratio. As $\hat{\boldsymbol{\Sigma}}_0$ and $\hat{\boldsymbol{\Sigma}}_1$ are diagonal matrices, $\lvert \hat{\boldsymbol{\Sigma}}_0 \rvert = \prod_{i=1}^{2} \hat{\sigma}_{i_0}^m$ and $\lvert \hat{\boldsymbol{\Sigma}}_1 \rvert = \prod_{i=1}^{2} \hat{\sigma}_{i_1}^m$ which leads to \eqref{eq:glrt_detector}.

\section{Special cases of the GLRT detector}
\label{sec:appendixB}

The detector \eqref{eq:glrt_detector} can be declined in already-known cases.

\subsection{Single array case}
For a single array, we have  $\hat{\sigma}_{1_0}^2 = \cfrac{\mathbf{x}_1^H \mathbf{M}_{11}^{-1} \mathbf{x}_1}{m}$ and
$\hat{\sigma}_{1_1}^2 = \cfrac{1}{m} \left( \mathbf{x}_1^H \mathbf{M}_{11}^{-1} \mathbf{x}_1 - \cfrac{\lvert \mathbf{p}_1^H \mathbf{M}_{11}^{-1} \mathbf{x}_1 \rvert ^2}{\mathbf{p}_1^H \mathbf{M}_{11}^{-1} \mathbf{p}_1} \right)$.

Replacing into the likelihood $L_G(\mathbf{x}) = \cfrac{\lvert \hat{\boldsymbol{\Sigma}}_0 \rvert ^2}{\lvert \hat{\boldsymbol{\Sigma}}_1 \rvert ^2} 	= \cfrac{\hat{\sigma}_{1_0}^{2m}}{\hat{\sigma}_{1_1}^{2m}}$ leads to:
\begin{align*}
	L_G(\mathbf{x})^{1/m} 
	& = \cfrac{\hat{\sigma}_{1_0}^2}{\hat{\sigma}_{1_1}^2} 
	= \cfrac{\cfrac{\mathbf{x}_1^H \mathbf{M}_{11}^{-1} \mathbf{x}_1}{m}}{\cfrac{1}{m} \left( \mathbf{x}_1^H \mathbf{M}_{11}^{-1} \mathbf{x}_1 - \cfrac{\lvert \mathbf{p}_1^H \mathbf{M}_{11}^{-1} \mathbf{x}_1 \rvert ^2}{\mathbf{p}_1^H \mathbf{M}_{11}^{-1} \mathbf{p}_1} \right)} \, , \\
	& = \left( 1 - \cfrac{\lvert \mathbf{p}_1^H \mathbf{M}_{11}^{-1} \mathbf{x}_1 \rvert ^2}{\left( \mathbf{p}_1^H \mathbf{M}_{11}^{-1} \mathbf{p}_1 \right)\left( \mathbf{x}_1^H \mathbf{M}_{11}^{-1} \mathbf{x}_1 \right)} \right)^{-1}\, .
\end{align*}

Defining $l_G(\mathbf{x}) = \cfrac{L_G(\mathbf{x})^{1/m} - 1}{L_G(\mathbf{x})^{1/m}} = 1- L_G(\mathbf{x})^{-1/m}$, we obtain the well-known NMF detector \cite{Scharf1994}:
\begin{equation}
l_G(\mathbf{x}) = \cfrac{\lvert \mathbf{p}_1^H \mathbf{M}_{11}^{-1} \mathbf{x}_1 \rvert ^2}{\left( \mathbf{p}_1^H \mathbf{M}_{11}^{-1} \mathbf{p}_1 \right)\left( \mathbf{x}_1^H \mathbf{M}_{11}^{-1} \mathbf{x}_1 \right)} \, .
\end{equation}

\subsection{Uncorrelated arrays case}
When the two arrays are fully uncorrelated, we have $\hat{\sigma}_{i_0}^2 = 
\cfrac{\mathbf{x}_i^H \mathbf{M}_{ii}^{-1} \mathbf{x}_i}{m} $ and
$\hat{\sigma}_{i_1}^2 = 
\cfrac{1}{m} \left( \mathbf{x}_i^H \mathbf{M}_{ii}^{-1} \mathbf{x}_i - \cfrac{\lvert \mathbf{p}_i^H \mathbf{M}_{ii}^{-1} \mathbf{x}_i \rvert ^2}{\mathbf{p}_i^H \mathbf{M}_{ii}^{-1} \mathbf{p}_i} \right)
$. We obtain:
\begin{align}
	L_G(\mathbf{x}) 
	& = \cfrac{\lvert \hat{\boldsymbol{\Sigma}}_0 \rvert ^2}{\lvert \hat{\boldsymbol{\Sigma}}_1 \rvert ^2} 
	= \cfrac{\displaystyle \prod_{i=1}^{2} \hat{\sigma}_{i_0}^{2m}}{\displaystyle \prod_{i=1}^{2} \hat{\sigma}_{i_1}^{2m}}
	= \prod_{i=1}^{2} \left[ \cfrac{\hat{\sigma}_{i_0}^{2}}{\hat{\sigma}_{i_1}^{2}} \right]^m \, , \\
	& = \prod_{i=1}^{2} \left[ 1 - \cfrac{\lvert \mathbf{p}_i^H \mathbf{M}_{ii}^{-1} \mathbf{x}_i \rvert ^2}{\left( \mathbf{p}_i^H \mathbf{M}_{ii}^{-1} \mathbf{p}_i \right)\left( \mathbf{x}_i^H \mathbf{M}_{ii}^{-1} \mathbf{x}_i \right)} \right]^{-m}.
\end{align}
This corresponds to the MIMO ANMF detector on independent arrays presented in \cite{Chong201003}.

\subsection{$\boldsymbol{\Sigma}=\sigma \, \mathbf{I}_{2m}$ case}
When $\boldsymbol{\Sigma}$ is the identity matrix up to a scalar factor, $\sigma_1 = \sigma_2$, whose estimators are renamed $\hat{\sigma}_0$ under  $H_0$ and $\hat{\sigma}_1$ under $H_1$.
\begin{align*}
	L_G(\mathbf{x}) 
	& = \cfrac{\lvert \hat{\boldsymbol{\Sigma}}_0 \rvert ^2}{\lvert \hat{\boldsymbol{\Sigma}}_1 \rvert ^2}  \text{, with } 
\hat{\boldsymbol{\Sigma}}_0 = \hat{\sigma}_0 \mathbf{I}_{mL} \text{ and }
\hat{\boldsymbol{\Sigma}}_1 = \hat{\sigma}_1 \mathbf{I}_{mL} \, , \\
	& = \cfrac{\hat{\sigma}_0^{2mL}}{\hat{\sigma}_1^{2mL}} \, .
\end{align*}
From $\hat{\boldsymbol{\Sigma}}_0 = \text{Re} \left( \mathbf{M}^{-1} \hat{\boldsymbol{\Sigma}}_0^{-1} \mathbf{x} \, \mathbf{x}^H \right)$, we have the following relations:
\begin{align*}
	\tr \left( \hat{\boldsymbol{\Sigma}}_0 \right) & = \tr \left[ \text{Re} \left( \mathbf{M}^{-1} \hat{\boldsymbol{\Sigma}}_0^{-1} \mathbf{x} \,  \mathbf{x}^H \right) \right] \, , \\
	\hat{\sigma}_0 \tr \left( \mathbf{I}_{mL} \right) & = \cfrac{1}{\hat{\sigma}_0} \text{Re} \left[ \tr \left( \mathbf{M}^{-1} \mathbf{x} \, \mathbf{x}^H \right) \right] \, , \\
	\hat{\sigma}_0^2 & = \cfrac{\mathbf{x}^H \mathbf{M}^{-1} \mathbf{x}}{mL} \text{,  as } \mathbf{M}^{-1} \text{ is positive definite}  \, .
\end{align*}
Identically, we have:
\begin{align*}
\hat{\sigma}_1^2 & = 
\cfrac{ \left( \mathbf{x} - \mathbf{P} \boldsymbol{\alpha} \right)^H \mathbf{M}^{-1} \left( \mathbf{x} - \mathbf{P} \boldsymbol{\alpha} \right)}{mL} \, , \\ 
&= \cfrac{ \mathbf{x}^H \mathbf{M}^{-1} \mathbf{x} - \mathbf{x}^H \mathbf{M}^{-1} \mathbf{P} \left( \mathbf{P}^H \mathbf{M}^{-1} \mathbf{P} \right)^{-1} \mathbf{P}^H \mathbf{M}^{-1} \mathbf{x} }{mL} \, .
\end{align*}
\begin{align*}
	L_G & (\mathbf{x})^{1/mL} = \cfrac{\hat{\sigma}_0^{2}}{\hat{\sigma}_1^{2}} \, , \\
	& = \cfrac{\mathbf{x}^H \mathbf{M}^{-1} \mathbf{x}}{ \mathbf{x}^H \mathbf{M}^{-1} \mathbf{x} - \mathbf{x}^H \mathbf{M}^{-1} \mathbf{P} \left( \mathbf{P}^H \mathbf{M}^{-1} \mathbf{P} \right)^{-1} \mathbf{P}^H \mathbf{M}^{-1} \mathbf{x} } \, ,\\
	& = \left( 1 - \cfrac{\mathbf{x}^H \mathbf{M}^{-1} \mathbf{P} \left( \mathbf{P}^H \mathbf{M}^{-1} \mathbf{P} \right)^{-1} \mathbf{P}^H \mathbf{M}^{-1} \mathbf{x} }{\mathbf{x}^H \mathbf{M}^{-1} \mathbf{x}} \right)^{-1} \, .
\end{align*}
By defining $l_G(\mathbf{x}) = \cfrac{L_G(\mathbf{x})^{1/mL} - 1}{L_G(\mathbf{x})^{1/mL}}$ or $L_G(\mathbf{x})^{1/mL} = \left[ 1 - l_G(\mathbf{x}) \right] ^{-1}$, we obtain an equivalent test:
\begin{equation}
	l_G(\mathbf{x}) = \cfrac{\mathbf{x}^H \mathbf{M}^{-1} \mathbf{P} \left( \mathbf{P}^H \mathbf{M}^{-1} \mathbf{P} \right)^{-1} \mathbf{P}^H \mathbf{M}^{-1} \mathbf{x} }{\mathbf{x}^H \mathbf{M}^{-1} \mathbf{x}} \, ,
\end{equation}
which corresponds to the subspace version of the ACE test presented in \cite{Raghavan2016}.

\section{Rao's detector derivation}
\label{sec:appendixC}

The partial derivative of the log-likelihood function is defined as:
\begin{equation*}
    \cfrac{\partial \ln p_{\mathbf{x}}(\mathbf{x}; \boldsymbol{\xi}_R, \boldsymbol{\xi}_S)}{\partial \boldsymbol{\xi}_R} = 
    \begin{bmatrix} 
    \cfrac{\partial \ln p_{\mathbf{x}}(\mathbf{x}; \boldsymbol{\xi}_R, \boldsymbol{\xi}_S)}{\partial  \mathrm{Re} \left( \boldsymbol{\alpha} \right)} \\ 
    \cfrac{\partial \ln p_{\mathbf{x}}(\mathbf{x}; \boldsymbol{\xi}_R, \boldsymbol{\xi}_S)}{\partial  \mathrm{Im} \left( \boldsymbol{\alpha} \right)} 
    \end{bmatrix}.
\end{equation*}
From \cite{Kay1993} (15.60), we obtain: 
\begin{equation*}
    \cfrac{\partial \ln p_{\mathbf{x}}(\mathbf{x}; \boldsymbol{\xi}_R, \boldsymbol{\xi}_S)}{\partial \boldsymbol{\xi}_R} = 
    \begin{bmatrix} 
    2 \, \text{Re} \left[ \mathbf{P}^H \mathbf{C}^{-1} (\boldsymbol{\xi}_S) \left( \mathbf{x} - \mathbf{P} \boldsymbol{\alpha} \right) \right] \\[2mm] 
    2 \, \text{Im} \left[ \mathbf{P}^H \mathbf{C}^{-1} (\boldsymbol{\xi}_S) \left( \mathbf{x} - \mathbf{P} \boldsymbol{\alpha} \right) \right] 
    \end{bmatrix} \, .
\end{equation*}
Thus:
\begin{equation}
    \left. 
    \cfrac{\partial \ln p_{\mathbf{x}}(\mathbf{x}; \boldsymbol{\xi}_R, \boldsymbol{\xi}_S)}{\partial \boldsymbol{\xi}_R} \right| 
    _{\begin{matrix}
    \boldsymbol{\xi}_R = \mathbf{0}_{\phantom{R_0}} \\
    \boldsymbol{\xi}_S = \hat{\boldsymbol{\xi}}_{S_0}
    \end{matrix}} = 
    \begin{bmatrix} 
    2 \, \text{Re} \left[ \mathbf{P}^H \mathbf{C}^{-1} (\hat{\boldsymbol{\xi}}_{S_0}) \mathbf{x} \right] \\[2mm]
    2 \, \text{Im} \left[ \mathbf{P}^H \mathbf{C}^{-1} (\hat{\boldsymbol{\xi}}_{S_0}) \mathbf{x} \right] 
    \end{bmatrix} \, .
    \label{eq:rao_derivee_partielle}
\end{equation}
Using \cite{Kay1993} (15.52) $\mathbf{I}_{\boldsymbol{\xi}_S \boldsymbol{\xi}_R}(\boldsymbol{\xi}_R, \boldsymbol{\xi}_S) = \mathbf{I}_{\boldsymbol{\xi}_R \boldsymbol{\xi}_S}(\boldsymbol{\xi}_R, \boldsymbol{\xi}_S) = \mathbf{0}$, thus $\left[ \mathbf{I}^{-1} (\boldsymbol{\xi}_R, \boldsymbol{\xi}_S) \right]_{\boldsymbol{\xi}_R \boldsymbol{\xi}_R} = \mathbf{I}_{\boldsymbol{\xi}_R \boldsymbol{\xi}_R}^{-1} (\boldsymbol{\xi}_R, \boldsymbol{\xi}_S)$.
The elements of $\mathbf{I}_{\boldsymbol{\xi}_R \boldsymbol{\xi}_R} (\boldsymbol{\xi}_R, \boldsymbol{\xi}_S)$ are given by:
\begin{eqnarray*}
    \left[ \mathbf{I}_{\boldsymbol{\xi}_R \boldsymbol{\xi}_R} (\boldsymbol{\xi}_R, \boldsymbol{\xi}_S) \right]_{11} 
    & = & 
    2 \, \mathbf{P}^H \mathbf{C}^{-1} (\boldsymbol{\xi}_S) \mathbf{P} \, , \\
    \left[ \mathbf{I}_{\boldsymbol{\xi}_R \boldsymbol{\xi}_R} (\boldsymbol{\xi}_R, \boldsymbol{\xi}_S) \right]_{22} 
    & = & 
    2 \, \mathbf{P}^H \mathbf{C}^{-1} (\boldsymbol{\xi}_S) \mathbf{P} \, , \\
    \left[ \mathbf{I}_{\boldsymbol{\xi}_R \boldsymbol{\xi}_R} (\boldsymbol{\xi}_R, \boldsymbol{\xi}_S) \right]_{12} 
    & = & 
    2 \, \text{Re} \left[ \mathrm{i} \, \mathbf{P}^H \mathbf{C}^{-1} (\boldsymbol{\xi}_S) \mathbf{P} \right]  = \mathbf{0} \, , \\
    \left[ \mathbf{I}_{\boldsymbol{\xi}_R \boldsymbol{\xi}_R} (\boldsymbol{\xi}_R, \boldsymbol{\xi}_S) \right]_{21} 
    & = & 
    2 \, \text{Re} \left[ - \mathrm{i} \, \mathbf{P}^H \mathbf{C}^{-1} (\boldsymbol{\xi}_S) \mathbf{P} \right]  = \mathbf{0} \, .
\end{eqnarray*}
We have therefore:
\begin{equation}
    \mathbf{I}_{\boldsymbol{\xi}_R \boldsymbol{\xi}_R} (\boldsymbol{\xi}_R, \boldsymbol{\xi}_S) = 
    \begin{bmatrix} 
    2 \, \mathbf{P}^H \mathbf{C}^{-1} (\boldsymbol{\xi}_S) \mathbf{P} & \mathbf{0} 
    \\ 
    \mathbf{0} & 2 \, \mathbf{P}^H \mathbf{C}^{-1} (\boldsymbol{\xi}_S) \mathbf{P}
    \end{bmatrix}.
    \label{eq:rao_FIM}
\end{equation}
Finally, by replacing \eqref{eq:rao_derivee_partielle} and \eqref{eq:rao_FIM} into \eqref{eq:rao_formula} leads to:
\footnotesize
\begin{equation*}
L_R(\mathbf{x}) = 
\begin{bmatrix} 
2 \, \text{Re} \left( \mathbf{x}^H \mathbf{C}^{-1} (\hat{\boldsymbol{\xi}}_{S_0}) \mathbf{P} \right) & -2 \, \text{Im} \left( \mathbf{x}^H \mathbf{C}^{-1} (\hat{\boldsymbol{\xi}}_{S_0}) \mathbf{P} \right) 
\end{bmatrix} 
\begin{bmatrix} 
	\cfrac{\left( \mathbf{P}^H \mathbf{C}^{-1} (\hat{\boldsymbol{\xi}}_{S_0}) \mathbf{P} \right) ^{-1}}{2} & \mathbf{0} \\ \mathbf{0} & \cfrac{\left( \mathbf{P}^H \mathbf{C}^{-1} (\hat{\boldsymbol{\xi}}_{S_0}) \mathbf{P} \right) ^{-1}}{2}
\end{bmatrix} 
\begin{bmatrix} 
2 \, \text{Re} \left( \mathbf{P}^H \mathbf{C}^{-1} (\hat{\boldsymbol{\xi}}_{S_0}) \mathbf{x} \right) \\
2 \, \text{Im} \left( \mathbf{P}^H \mathbf{C}^{-1} (\hat{\boldsymbol{\xi}}_{S_0}) \mathbf{x} \right) 
\end{bmatrix} ,
\end{equation*}
\normalsize
which simplifies to:
\begin{align*}
L_R(\mathbf{x}) = 2 \left[ \text{Re} \left( \mathbf{x}^H \mathbf{C}^{-1} (\hat{\boldsymbol{\xi}}_{S_0}) \mathbf{P} \right) \left( \mathbf{P}^H \mathbf{C}^{-1} (\hat{\boldsymbol{\xi}}_{S_0}) \mathbf{P} \right)^{-1}
\right. 
\text{Re} \left( \mathbf{P}^H \mathbf{C}^{-1} (\hat{\boldsymbol{\xi}}_{S_0}) \mathbf{x} \right) \\ - \text{Im} \left( \mathbf{x}^H \mathbf{C}^{-1} (\hat{\boldsymbol{\xi}}_{S_0}) \mathbf{P} \right) 
\left( \mathbf{P}^H \mathbf{C}^{-1} (\hat{\boldsymbol{\xi}}_{S_0}) \mathbf{P} \right)^{-1} \left. \text{Im} \left( \mathbf{P}^H \mathbf{C}^{-1} (\hat{\boldsymbol{\xi}}_{S_0}) \mathbf{x} \right) \right] .
\end{align*}

Knowing that $\mathbf{P}^H \mathbf{C}^{-1} (\hat{\boldsymbol{\xi}}_{S_0}) \mathbf{P}$ is real and positive definite, it can be factorized and incorporated into the real and imaginary parts. After some algebraic manipulation, we obtain \eqref{eq:rao_dector}.

\section{Maximum likelihood estimator of the covariance matrix in Compound-Gaussian clutter}
\label{sec:appendixD}

The likelihood of $\displaystyle\left\{\mathbf{x}_k\right\}_{k\in[1,K]}$ under $H_0$ can  be rewritten as:
\begin{align*}
p_{\mathbf{x}}&\left(\left\{\mathbf{x}_k\right\}_{k\in[1,K]} ; \mathbf{M}, \mathbf{T}_k, H_0\right) \\ & = \prod_{k=1}^{K} \cfrac{1}{\pi^{2m} \lvert \widetilde{\mathbf{C}} \rvert} \exp \left( - \mathbf{x}_k^H \widetilde{\mathbf{C}}^{-1} \mathbf{x}_k \right) \, ,\\
	& = \cfrac{1}{\pi^{2m\,K} \lvert \mathbf{M} \rvert ^K} \prod_{k=1}^{K} \cfrac{1}{\lvert \mathbf{T}_k \rvert ^2} \exp \left( - \mathbf{x}_k^H \mathbf{T}_k^{-1} \mathbf{M}^{-1} \mathbf{T}_k^{-1} \mathbf{x}_k \right)\, ,
\end{align*}
where $\widetilde{\mathbf{C}} = \mathbf{T}_k \mathbf{M} \mathbf{T}_k$, and	$\mathbf{T}_k = \begin{bmatrix} \sqrt{\tau_{1_k}} & 0 \\ 0 & \sqrt{\tau_{2_k}} \\ \end{bmatrix} \otimes \mathbf{I}_m$.

The log-likelihood can be written as:
\begin{multline}
	\ln p_{\mathbf{x}}\left(\left\{\mathbf{x}_k\right\}_{k\in[1,K]}; \mathbf{M}, \mathbf{T}_k, H_0\right) = -2m \, K \ln\pi -K \, \ln \lvert \mathbf{M} \rvert \\ + 2 \sum_{k=1}^K \ln \lvert \mathbf{T}_k^{-1} \rvert
	- \sum_{k=1}^K \mathbf{x}_k^H \mathbf{T}_k^{-1} \mathbf{M}^{-1} \mathbf{T}_k^{-1} \mathbf{x}_k \, .
	\label{eq:log-vraisamblance_2_ss_antennes}
\end{multline}

According to \cite{Petersen2012} (82), the derivative with respect to $\mathbf{T}_k^{-1}$ is:
\begin{equation}
	\cfrac{\partial \ln p_{\mathbf{x}}(\left\{\mathbf{x}_k\right\}_{k\in[1,K]}; \mathbf{M}, \mathbf{T}_k, H_0)}{\partial \mathbf{T}_k^{-1}} 
	= 2 \,\mathbf{T}_k \\ - 2 \,\text{Re} \left( \mathbf{M}^{-1} \mathbf{T}_k^{-1} \mathbf{x}_k \, \mathbf{x}_k^H \right)\, .
\end{equation}

Following the same approach as in Appendix \ref{sec:appendixA}, we obtain the minimum for $\mathbf{T}_k$ for a fixed $\mathbf{M}$:
\begin{equation}
	\widehat{\mathbf{T}}_k = \text{Re} \left( \mathbf{M} ^{-1} \widehat{\mathbf{T}}_k ^{-1} \mathbf{x}_k \, \mathbf{x}_k^H \right) \, ,
\end{equation}
where
\begin{equation}
	\widehat{\mathbf{T}}_k = \begin{bmatrix} \sqrt{\hat{\tau}_{1_k}} & 0 \\ 0 & \sqrt{\hat{\tau}_{2_k}} \\ \end{bmatrix} \otimes \mathbf{I}_m \, ,
 \label{eq:T_k}
\end{equation}
and
\begin{eqnarray}
\hat{\tau}_{1_k} \!\!\!& = & \!\!\!t_1 + \sqrt{\cfrac{t_1}{t_2}} \, t_{12}\, , \\
\hat{\tau}_{2_k} \!\!\!& =  & \!\!\!t_2 + \sqrt{\cfrac{t_2}{t_1}} \, t_{12}\, ,
\end{eqnarray}
with 
\begin{eqnarray}
t_1 \!\!\!& = &\!\!\!\cfrac{1}{m} \, \mathbf{x}_{1,k}^H \mathbf{M}_{11}^{-1} \mathbf{x}_{1,k} \, ,\\
\!\!\!t_2 & = & \!\!\! \cfrac{1}{m} \, \mathbf{x}_{2,k}^H \mathbf{M}_{22}^{-1} \mathbf{x}_{2,k}\, ,\\
t_{12} \!\!\!& = & \!\!\!\cfrac{1}{m} \, \text{Re} \left( \mathbf{x}_{1,k}^H \mathbf{M}_{12}^{-1} \mathbf{x}_{2,k} \right)\, .
\end{eqnarray}

Replacing $\mathbf{T}_k$ by $\widehat{\mathbf{T}}_k$ in \eqref{eq:log-vraisamblance_2_ss_antennes} and deriving with respect to $\mathbf{M}^{-1}$ lead to:
\begin{equation}
	\cfrac{\partial \ln p_{\mathbf{x}}(\left\{\mathbf{x}_k\right\}_{k\in[1,K]}; \mathbf{M}, \widehat{\mathbf{T}}_k, H_0)}{\partial \mathbf{M}^{-1}}  =  K \, \mathbf{M} 
	 - \sum_{k=1}^K \left( \widehat{\mathbf{T}}_k^{-1} \mathbf{x}_k \right) \left( \widehat{\mathbf{T}}_k^{-1} \mathbf{x}_k \right)^H \, , 
\end{equation}
and the minimum in $\mathbf{M}$ is given by:
\begin{align}
	\widehat{\mathbf{M}} 
	& = \cfrac{1}{K} \sum_{k=1}^K \left( \widehat{\mathbf{T}}_k^{-1} \mathbf{x}_k \right) \left( \widehat{\mathbf{T}}_k^{-1} \mathbf{x}_k \right)^H \, , \\
	& = \cfrac{1}{K} \sum_{k=1}^K \widehat{\mathbf{T}}_k^{-1} \mathbf{x}_k \mathbf{x}_k^H \widehat{\mathbf{T}}_k^{-1} \, .
	\label{eq:Tyler}
\end{align}

The estimator \eqref{eq:Tyler} is independent of the textures. This could be shown by substituting $\mathbf{x}_k = \begin{bmatrix} \mathbf{x}_{1,k} \\ \mathbf{x}_{2,k} \end{bmatrix}$ by $\begin{bmatrix} \sqrt{\tau_{1_k}} \, \mathbf{c}_{1,k} \\ \sqrt{\tau_{2_k}} \, \mathbf{c}_{2,k} \end{bmatrix}$ in \eqref{eq:T_k} and \eqref{eq:Tyler}.

\bibliographystyle{IEEEbib}
\bibliography{references}

\end{document}